\documentclass[11pt,envcountsame]{article}
\usepackage{fullpage}
\usepackage{ucs}
\usepackage[utf8x]{inputenc}
\usepackage{hyperref}
\usepackage[T1]{fontenc}
\usepackage{listings}
\usepackage{microtype}
\usepackage{lmodern}
\usepackage{graphicx}
\usepackage{makeidx}
\usepackage{amssymb}
\usepackage{appendix}
\usepackage{pdfpages}

\usepackage{amsmath}
\usepackage{amsthm}
\usepackage{calc}
\usepackage{tabularx}
\usepackage{algpseudocode}
\usepackage{algorithm}
\usepackage{cleveref}
\usepackage{thmtools}
\usepackage{thm-restate}

\newtheorem{definition}{Definition}

\newtheorem{lemma}[definition]{Lemma}
\newtheorem{theorem}[definition]{Theorem}

\newtheorem{claim}[definition]{Claim}

\newtheorem{example}[definition]{Example}

\usepackage{newtxtext}
\usepackage[hmargin=1in,vmargin=1in]{geometry}
\usepackage{authblk}

\newcommand{\Roit}{\ensuremath{R_{\scriptscriptstyle 1/3}}}
\newcommand{\hroit}{\ensuremath{\hat R_{\scriptscriptstyle 1/3}}}
\newcommand{\roik}[1]{\ensuremath{R_{\scriptscriptstyle 1/#1}}}

\newcommand{\gadget}{\ensuremath{\Phi_1}}

\newcommand{\br}{\ensuremath{\mathit{BR}}}

\newcommand{\pr}{\ensuremath{\mathrm{pr}}} 

\newcommand{\clone}[1]{\ensuremath{[#1]}}
\newcommand{\cclone}[1]{\ensuremath{\langle #1 \rangle}}

\newcommand{\str}[1]{\ensuremath{\mathrm{Str}}}
\newcommand{\pcclone}[1]{\ensuremath{\langle #1 \rangle_{\not \exists}}}

\newcommand{\ar}{\ensuremath{\mathrm{ar}}}

\newcommand{\eq}{\ensuremath{\rm{Eq}}}
\newcommand{\ppol}{\ensuremath{\rm{pPol}}}
\newcommand{\pol}{\ensuremath{\rm{Pol}}}
\newcommand{\inv}{\ensuremath{\rm{Inv}}}

\newcommand{\domain}{\ensuremath{\mathrm{domain}}}

\newcommand{\CSP}{\protect\ensuremath\problemFont{CSP}}
\newcommand{\SAT}{\protect\ensuremath\problemFont{SAT}}
\newcommand{\problemFont}[1]{\textsc{#1}}

\newcommand{\rest}[1]{\ensuremath{#1_{\mid \mathbb{B}}}}
\newcommand{\sig}{\ensuremath{\mathrm{Sig}}}
\newcommand{\val}{\ensuremath{\mathrm{Val}}}
\newcommand{\seq}{\ensuremath{\mathrm{Seq}}}
\newcommand{\rel}[1]{\ensuremath{\Psi_{#1}}}

\newcommand{\malt}{\ensuremath{\phi}}
\newcommand{\infd}{\ensuremath{D_{\infty}}}

\begin{document}
\title{Kernelization of Constraint Satisfaction Problems: A Study
  through Universal Algebra \footnote{This is an extended preprint of
    {\em Kernelization of Constraint Satisfaction Problems: A Study
  through Universal Algebra}, appearing in Proceedings of the 23rd
International Conference on Principles and Practice of Constraint
Programming (CP 2017)}}
\author[1]{Victor Lagerkvist\thanks{victor.lagerqvist@tu-dresden.de}}
\author[2]{Magnus Wahlstr\"om\thanks{magnus.wahlstrom@rhul.ac.uk}}
\affil[1]{\small Institut f\"ur Algebra, TU Dresden, Dresden, Germany}
\affil[2]{\small Department of Computer Science, Royal Holloway,
  University of London, Great Britain}
\date{}
\pagenumbering{gobble} 
\maketitle

\abstract{A {\em kernelization} algorithm for a computational problem
is a procedure which compresses an instance into an equivalent
instance whose size is bounded with respect to a complexity
parameter. For
the Boolean satisfiability problem (SAT), and the constraint
satisfaction problem (CSP), there exist many results concerning upper
and lower bounds for kernelizability of specific problems, but it is safe to
say that we lack general methods to determine whether a given
SAT problem admits a kernel of a particular size. This could be
contrasted to the currently flourishing research program of
determining the classical complexity of finite-domain CSP problems, where
almost all non-trivial tractable classes have been identified with the
help of algebraic properties.  In this paper, we take an algebraic approach to the problem of
characterizing the kernelization limits of NP-hard SAT and CSP problems,
parameterized by the number of variables. Our main focus is on problems
admitting linear kernels, as has, somewhat surprisingly,
previously been shown to exist. We show that a CSP problem has
a kernel with $O(n)$ constraints if it can be embedded (via a domain extension)
into a CSP problem which is preserved by a Maltsev operation.
We also study extensions of this towards SAT and CSP problems with kernels with $O(n^c)$ 
constraints, $c>1$, based on embeddings into CSP problems 
preserved by a $k$-edge operation, $k \leq c+1$.
These results follow via a variant of the celebrated few subpowers algorithm.
In the complementary direction, we give indication that the Maltsev
condition might be a complete characterization of SAT problems with linear kernels,
by showing that an algebraic condition that is shared by all problems with a Maltsev embedding
is also necessary for the existence of a linear kernel unless NP $\subseteq$ co-NP/poly.

\pagenumbering{arabic}  

\section{Introduction}
\label{sec:intro}
{\em Kernelization} is a preprocessing technique based on reducing an
instance of a computationally hard problem in polynomial time to an equivalent instance,
a {\em kernel}, whose size is bounded by a function $f$ with respect
to a given complexity parameter. The function $f$ is referred to as
the {\em size} of the kernel, and if the size is polynomially bounded
we say that the problem admits a {\em polynomial kernel}. 
A classical example is 
\textsc{Vertex Cover}, which admits a kernel with
$2k$ vertices, where $k$ denotes the size of the cover~\cite{nemhauser1975}. 
Polynomial kernels are of great interest in parameterized complexity, 
as well as carrying practical significance in speeding up subsequent
computations (e.g., the winning contribution in the 2016 PACE
challenge for \textsc{Feedback Vertex Set} used a novel 
kernelization step as a key component (see \url{https://pacechallenge.wordpress.com/}).

When the complexity parameter is a size parameter, e.g.,
the number of variables $n$, then such a size reduction 
is also referred to as \emph{sparsification} (although
a sparsification is not always required to run in polynomial time). 
A prominent example is the famous \emph{sparsification lemma}
that underpins research into the Exponential Time Hypothesis~\cite{impagliazzo98},
which shows that for every $k$ there is a subexponential-time
reduction from $k$-SAT on $n$ variables to $k$-SAT on $O(n)$ clauses,
and hence $\tilde O(n)$ bits in size.
However, the super-polynomial running time is essential to this result. 
Dell and van Melkebeek~\cite{DellM14} showed that 
$k$-SAT cannot be kernelized even down to size $O(n^{k-\varepsilon})$,
and \textsc{Vertex Cover} cannot be kernelized to size $O(n^{2-\varepsilon})$, 
for any $\varepsilon > 0$ unless the polynomial hierarchy collapses
(in the sequel, we will make this assumption implicitly).
These results suggest that in general, polynomial-time sparsification
can give no non-trivial size guarantees. (Note that a 
kernel of size $O(n^k)$ for $k$-SAT is trivial.)
The first result to the contrary was by Bart Jansen (unpublished until recently~\cite{JansenP16MFCS}),
who observed that \textsc{1-in-$k$-SAT}
admits a kernel with at most $n$ constraints using Gaussian elimination. 
More surprisingly, Jansen and Pieterse~\cite{jansen15} showed that
the \textsc{Not-All-Equal $k$-SAT} problem admits a kernel
with $O(n^{k-1})$ constraints, 
improving on the trivial bound by a factor of $n$ and
settling an implicit open problem. 
In later research,
they improved and generalized the method, and also showed that the bound of 
$O(n^{k-1})$ is tight~\cite{JansenP16MFCS}.
These improved upper bounds are all based on rephrasing the $\SAT$ problem
as a problem of low-degree polynomials, and exploiting linear dependence
to eliminate superfluous constraints.
Still, it is fair to say that 
we currently lack the tools 
for making a general analysis of the 
kernelizability of a generic $\SAT$ problem.

In this paper we take a step in this direction,
by studying the kernelizability of the {\em constraint
  satisfaction problem} over a constraint language $\Gamma$
($\CSP(\Gamma)$), parameterized by the number of variables $n$, which
can be viewed as the problem of 
determining whether a set of constraints over $\Gamma$
is satisfiable. Some notable examples of problems of this
kind are $k$-colouring, $k$-SAT, 1-in-$k$-SAT, and
not-all-equal-$k$-SAT. We will occasionally put a particular emphasis
on the Boolean $\CSP$ problem and therefore denote this problem by $\SAT(\Gamma)$. Note that $\CSP(\Gamma)$ has
a trivial polynomial kernel for any finite language $\Gamma$
(produced by simply discarding duplicate constraints), but the question
remains for which languages $\Gamma$ we can improve upon this. 
Concretely, our question in this paper is for which languages $\Gamma$
the problem $\CSP(\Gamma)$ admits a kernel of $O(n^c)$ constraints, 
for some $c \geq 1$, with a particular focus on linear kernels ($c=1$).

\textbf{The algebraic approach in parameterized and fine-grained complexity.}
For any language $\Gamma$, 
the classical complexity of $\CSP(\Gamma)$ (i.e., whether
$\CSP(\Gamma)$ is in P) is determined by the existence of certain
algebraic invariants of $\Gamma$ known as \emph{polymorphisms}~\cite{jeavons1998}. 
This gave rise to the \emph{algebraic approach} to characterizing the 
complexity of $\CSP(\Gamma)$ by studying algebraic properties.
It has been conjectured that for every $\Gamma$, $\CSP(\Gamma)$
is either in P or NP-complete, and that the tractability of a CSP problem 
can be characterized by a finite list of polymorphisms~\cite{bulatov2005}.
Recently, several independent results appeared, claiming to settle
this conjecture in the positive~\cite{bulatov2017,feder2017,zhuk2017}.

However, for purposes of parameterized and fine-grained complexity
questions, looking at polymorphisms alone is too coarse.
More technically, the polymorphisms of $\Gamma$ characterize 
the expressive power of $\Gamma$ up to \emph{primitive positive definitions},
i.e., up to the use of conjunctions, equality constraints, 
and existential quantification, whereas for many questions a liberal use
of existentially quantified local variables is not allowed. 
In such cases, one may look at the expressive power under 
 \emph{quantifier-free} primitive positive definitions (qfpp-definitions),
allowing only conjunctions and equality constraints. 
This expressive power is characterized by more fine-grained
algebraic invariants called \emph{partial polymorphisms}. 
For example, there are numerous dichotomy results 
for the complexity of \emph{parameterized} $\SAT(\Gamma)$
and $\CSP(\Gamma)$ problems, both for so-called FPT algorithms
and for kernelization~\cite{KratschMW16,KratschW10,KrokhinM12,Marx05},
and in each of the cases listed, a dichotomy is given which
is equivalent to requiring a finite list of partial polymorphisms of $\Gamma$. 
Similarly, Jonsson et al.~\cite{jonsson2017} showed that the
exact running times of NP-hard SAT$(\Gamma)$ and CSP$(\Gamma)$ problems
in terms of the number of variables $n$
are characterized by the partial polymorphisms of $\Gamma$. 

Unfortunately, studying properties of $\SAT(\Gamma)$ and $\CSP(\Gamma)$
for questions phrased in terms of the size parameter $n$
is again more complicated than for more permissive parameters $k$.
For example, it is known that for every finite set $P$ of 
strictly partial polymorphisms, the number of relations
invariant under $P$ is double-exponential in terms of
the arity $n$ (hence they cannot all be described
in a polynomial number of bits)~\cite[Lemma~35]{lagerkvist2016b}.
It can similarly be shown that the existence of a polynomial kernel
cannot be characterized by such a finite set $P$. 
Instead, such a characterization must be given in another way
(for example, Lagerkvist et al.~\cite{lagerkvist2015}
provide a way to finitely characterize all partial polymorphisms
of a finite Boolean language $\Gamma$). 

\textbf{Our results.} In Section~\ref{sec:embeddings} we generalize and extend the results of
Jansen and Pieterse~\cite{JansenP16MFCS} in the case of linear kernels
to a general recipe for NP-hard $\SAT$ and $\CSP$ problems 
in terms of the existence of a \emph{Maltsev embedding}, i.e.,
an embedding of a language $\Gamma$ into a 
tractable language $\Gamma'$ on a larger domain with a \emph{Maltsev polymorphism}. 
We show that for any language $\Gamma$ with a Maltsev embedding into
a finite domain, $\CSP(\Gamma)$ has a kernel with $O(n)$ constraints. More generally, we in
Section~\ref{sec:higher}, turn to the problem of finding kernels with
$O(n^c)$ constraints ($c > 1$) where we utilize {\em $k$-edge} embeddings,
and a technique which encompasses the recent results from Jansen and
Pieterse, concerning SAT problems representable as low-degree
polynomials over a finite field~\cite{JansenP16MFCS}. 
Attempting an algebraic characterization, we in Section~\ref{section:lower_bounds} also show
an infinite family of \emph{universal} partial operations
which are partial polymorphisms of every language $\Gamma$
with a Maltsev embedding, and show that these operations
guarantee the existence of a Maltsev embedding for $\Gamma$,
albeit into a language with an infinite domain. 

Turning to lower bounds against linear kernels, 
we show that the smallest of these universal partial operations
is also necessary, in the sense that for any Boolean language $\Gamma$
which is not invariant under this operation,
$\SAT(\Gamma)$ admits no kernel of size $O(n^{2-\varepsilon})$
for any $\varepsilon>0$. 
We conjecture that this can be completed into a tight characterization
-- i.e., that at least for Boolean languages $\Gamma$,  
$\SAT(\Gamma)$ admits a linear kernel
if and only if it is invariant under all universal partial 
Maltsev operations. 

\section{Preliminaries}
\label{sec:prel}
In this section we introduce the constraint satisfaction problem,
kernelization, and the algebraic machinery that will be used
throughout the paper. 

\subsection{Operations and Relations}
\label{sec:operations}
An $n$-ary function $f : D^{n} \rightarrow D$ over a domain $D$ is
typically referred to as a {\em operation} on $D$, although we will
sometimes use the terms function and operation interchangeably. We let $\ar(f) = n$
denote the arity of $f$. Similarly, if $R \subseteq D^{n}$ is an
$n$-ary relation over $D$ we let $\ar(R) = n$. If $t \in D^{n}$ is a
tuple we let $t[i]$ denote the $i$th element in $t$ and we let
$\pr_{i_1, \ldots, i_{n'}}(t) = (t[i_1], \ldots, t[i_{n'}])$, $n' \leq n$,
denote the {\em projection} of $t$ on (not necessarily distinct) coordinates $i_1, \ldots,
i_{n'} \in \{1, \ldots, n\}$. Similarly, if $R$ is an $n$-ary relation
we let $\pr_{i_1, \ldots, i_{n'}}(R) = \{\pr_{i_1, \ldots, i_{n'}}(t)
\mid t \in R\}$. We will often represent relations by logical
formulas, and if $\psi$ is a first-order formula with free variables
$x_1, \ldots, x_k$ we by $R(x_1, \ldots, x_k) \equiv \psi(x_1, \ldots, x_k)$
denote the relation $R = \{(f(x_1), \ldots, f(x_k)) \mid f \text{ is a
satisfying assignment to } \psi\}$.

\subsection{The Constraint Satisfaction Problem}
\label{sec:csp}
A set of relations $\Gamma$ is referred to as a {\em constraint
  language}. The {\em constraint satisfaction problem} over a
constraint language $\Gamma$ over $D$
($\CSP(\Gamma)$) is the computational decision problem defined as
follows.
\smallskip

\noindent
{\sc Instance:} A set $V$ of variables and a set $C$ of constraint
applications $R(v_1,\ldots,v_{k})$ where $R \in \Gamma$, $\ar(R) = k$,
and $v_1,\ldots,v_{k} \in V$.

\noindent
{\sc Question:} Is there a function $f : V \rightarrow D$ such
that $(f(v_1),\ldots,f(v_{k})) \in R$ for each 
$R(v_1,\ldots,v_{k})$~in~$C$?

\smallskip
In the particular case when $\Gamma$ is Boolean we denote
$\CSP(\Gamma)$ by $\SAT(\Gamma)$, and we let $\br$
denote the set of all Boolean relations. As an example, first consider the ternary relation $\Roit =
\{(0,0,1),(0,1,0),(1,0,0)\}$. It is then readily seen that
$\SAT(\{\Roit\})$ can be viewed as an alternative formulation of 
the 1-in-3-SAT problem restricted to instances consisting only of positive
literals. More generally, if we let $\roik{k} = \{(x_1, \ldots, x_k)
\in \{0,1\}^{k} \mid x_1 + \ldots + x_k = 1\}$, then
$\SAT(\{\roik{k}\})$ is a natural formulation of 1-in-$k$-SAT without
negation.

\subsection{Kernelization}
\label{sec:kern}
A {\em parameterized problem} is a subset of $\Sigma^* \times
\mathbb{N}$ where $\Sigma$ is a finite alphabet. Hence, each instance
is associated with a natural number, called the {\em parameter}.

\begin{definition}
  A {\em kernelization algorithm}, or a {\em kernel}, for a parameterized problem $L
  \subseteq \Sigma^{*} \times \mathbb{N}$ is a polynomial-time
  algorithm which, given an instance $(x,k) \in \Sigma^{*} \times \mathbb{N}$, computes $(x',k')
  \in \Sigma^{*}\times \mathbb{N}$ such that (1) $(x,k) \in L$ if and
  only if $(x',k') \in L$ and (2) $|x'| + k' \leq f(k)$ for some
  function $f$.
\end{definition}
The function $f$ in the above definition is sometimes called the {\em
size} of the kernel. In this paper, we are mainly 
interested in the case where the parameter denotes the number of
variables in a given instance. 

\subsection{Polymorphisms and Partial Polymorphisms}
\label{sec:pol}
In this section we define the link between constraint languages and
algebras that was promised in Section~\ref{sec:intro}. 
If $f$ is an $n$-ary operation and $t_1, \ldots, t_n$ a sequence of
$k$-ary tuples we can in a natural way obtain a $k$-ary tuple by
applying $f$ componentwise, i.e., $f(t_1, \ldots, t_n) = (f(t_1[1],
\ldots, t_n[1]), \ldots, f(t_1[k], \ldots, t_n[k]))$.

\begin{definition}
An $n$-ary operation $f$ is a {\em polymorphism} of a $k$-ary relation $R$
if $f(t_1, \ldots, t_n) \in R$ for each sequence of tuples $t_1, \ldots,
t_n \in R$. 
\end{definition}

If $f$ is a polymorphism of $R$ we also say that $R$ is {\em invariant
} under $f$, or that $f$ {\em preserves} $R$, and for a constraint language $\Gamma$ we let $\pol(\Gamma)$ denote the set of operations
preserving every relation in $\Gamma$.
Similarly, if $F$ is
a set of functions, we let $\inv(F)$ denote the set of all relations
invariant under $F$.
Sets of functions of the form $\pol(\Gamma)$ are referred to as {\em
  clones}. It is well known that $\pol(\Gamma)$ (1) for each $n \geq
1$ and each $1 \leq i \leq n$ contains the {\em projection function}
$\pi^n_i(x_1, \ldots, x_i, \ldots, x_n) = x_i$, and (2) if $f, g_1,
\ldots, g_{m} \in \pol(\Gamma)$, where $\ar(f) = m$ and all $g_i$ have the same arity
$n$, then the {\em composition} $f \circ g_1, \ldots, g_{m}(x_1,
\ldots, x_n) = f(g_1(x_1, \ldots, x_n), \ldots, g_m(x_1, \ldots,
x_n))$ is included in $\pol(\Gamma)$. Similarly, sets of the form
$\inv(F)$ are referred to as {\em relational clones}, or {\em
  co-clones}, and are sets of relations closed under {\em primitive
  positive definitions} (pp-definitions), which are logical formulas
consisting of existential quantification, conjunction, and equality
constraints. In symbols, we say that a $k$-ary relation $R$ has a
pp-definition over a constraint language $\Gamma$ over a domain $D$ if
$R(x_1, \ldots, x_{k}) \equiv \exists y_1, \ldots,
y_{k'}\, . \, R_1(\mathbf{x_1}) \wedge \ldots \wedge
R_m({\mathbf{x_m}}),$ where each $R_i \in \Gamma \cup \{\eq\}$, $\eq =
\{(x,x) \mid x \in D\}$ and
each $\mathbf{x_i}$ is an $\ar(R_i)$-ary tuple of variables over
$x_1,\ldots, x_{k}$, $y_1, \ldots, y_{k'}$. Clones and co-clones are
related via the following {\em Galois connection}.

\begin{theorem}[\cite{BKKR69i,BKKR69ii,Gei68}]
  \label{theorem:galois}
  Let $\Gamma$ and $\Gamma'$ be two constraint languages. Then $\Gamma \subseteq \inv(\pol(\Gamma'))$ if and only if  $\pol(\Gamma') \subseteq
    \pol(\Gamma)$.
\end{theorem}

As a shorthand, we let $\clone{F} = \pol(\inv(F))$
denote the smallest clone containing $F$ and $\cclone{\Gamma} =
\inv(\pol(\Gamma))$ the smallest co-clone containing $\Gamma$.
Using Theorem~\ref{theorem:galois} Jeavons et al.\ proved that if
$\Gamma$ and $\Gamma'$ are two finite constraint languages and
$\pol(\Gamma) \subseteq \pol(\Gamma')$, then $\CSP(\Gamma')$ is
polynomial-time many-one reducible to
$\CSP(\Gamma)$~\cite{jeavons1998}. As remarked in
Section~\ref{sec:intro}, while this theorem is useful for establishing
complexity dichotomies for $\CSP$ and related
problems~\cite{barto2014,creignou2008b}, it offers little information
on whether a problem admits a kernel of a particular size.
Hence, in order to have any hope of studying kernelizability of $\SAT$
problems, we need algebras more fine-grained than polymorphisms. In
our case these algebras will consist of {\em partial operations}
instead of total operations. An $n$-ary partial operation over a set $D$ of values is a
map of the form $f : X \rightarrow D$, where $X \subseteq D^{n}$ is
called the {\em domain} of $f$. As in the case of total operations we
let $\ar(f) = n$, and furthermore let $\domain(f) = X$. If $f$ and $g$ are $n$-ary partial
operations such that $\domain(g) \subseteq \domain(f)$ and $f(x_1,
\ldots, x_n) = g(x_1, \ldots, x_n)$ for each $(x_1, \ldots, x_n) \in
\domain(g)$, then $g$ is said to be a {\em subfunction} of $f$. 
\begin{definition}
An $n$-ary partial operation $f$ is a {\em partial polymorphism} of a $k$-ary relation $R$ if, for
every sequence $t_1, \ldots, t_n \in R$, either $f(t_1, \ldots, t_n)
\in R$ or there exists $i \in \{1, \ldots, k\}$ such that $(t_1[i],
\ldots, t_n[i]) \notin \domain(f)$. 
\end{definition}
Again, this notion easily
generalizes to constraint languages, and if we let $\ppol(\Gamma)$ denote
the set of partial polymorphisms of the constraint language $\Gamma$,
we obtain a {\em strong partial clone}. It is known that strong
partial clones are sets of partial operations which are (1) closed
under composition of partial operations and (2) containing all partial
projection functions~\cite{romov1981}. More formally, the first
condition means that if $f, g_1, \ldots, g_m$ are included in the
strong partial clone, where $f$ is $m$-ary
and every $g_i$ is $n$-ary, then the function $f \circ g_1, \ldots, g_m(x_1, \ldots,
x_n) = f(g_1(x_1, \ldots, x_n), \ldots, g_m(x_1, \ldots, x_n))$ is
also included in the strong partial clone, and this function will be
defined for $(x_1, \ldots, x_n) \in D^{n}$ if and only if $(x_1,
\ldots, x_n) \in \bigcap^{m}_{i=1} \domain(g_i)$ and $(g_1(x_1,
\ldots, x_n), \ldots, g_m(x_1, \ldots, x_n)) \in \domain(f)$. The
second condition, containing all partial projection functions, is
known to be equivalent to closure under taking subfunctions;
a property which in the literature is sometimes called {\em strong}.

If $F$ is a set of partial functions we let $\inv(F)$ denote
the set of all relations invariant under $F$, but this time $\inv(F)$
is in general not closed under pp-definitions, but under {\em
  quantifier-free primitive positive definitions}
(qfpp-definitions). As the terminology suggests, a relation $R$ has a
qfpp-definition over $\Gamma$ if $R$ is definable via a pp-formula
which does not make use of existential quantification. Such formulas
are sometimes simply called {\em conjunctive formulas}. We have the
following Galois connection.

\begin{theorem}[\cite{Gei68,romov1981}]
  \label{theorem:pgalois}
  Let $\Gamma$ and $\Gamma'$ be two constraint languages. Then 
  $\Gamma \subseteq \inv(\ppol(\Gamma'))$ if and only if  $\ppol(\Gamma') \subseteq
    \ppol(\Gamma)$.
\end{theorem}

As a shorthand we let $\inv(\ppol(\Gamma)) = \pcclone{\Gamma}$.
Jonsson et al.~\cite{jonsson2017} showed the following useful theorem.

\begin{theorem}\label{theorem:ppol_red}
Let $\Gamma$ and $\Gamma'$ be two finite constraint languages. If
$\ppol(\Gamma) \subseteq \ppol(\Gamma')$ then there exists a
polynomial-time many-one reduction from $\SAT(\Gamma')$ to
$\SAT(\Gamma)$ which maps an instance $(V,C)$ of $\SAT(\Gamma')$ to an
instance 
$(V',C')$ of $\SAT(\Gamma)$ where $|V'| \leq |V|$ and
$|C'| \leq c|C|$, 
where $c$ depends only on $\Gamma$ and $\Gamma'$.
\end{theorem}

\subsection{Maltsev Operations, Signatures and Compact Representations}
\label{sec:kedge}
A {\em Maltsev operation} over $D \supseteq \{0,1\}$ is a ternary
operation $\malt$ which for all $x,y \in D$ satisfies the two identities
$\malt(x,x,y) = y$ and $\malt(x,y,y) = x$. Before we can explain the
powerful, structural properties of relations invariant under Maltsev
operations, we need a few technical definitions from Bulatov and Dalmau~\cite{bulatov2006b}.
Let $t,t'$ be two $n$-ary tuples over $D$. We say that $(t,t')$ {\em witnesses} a tuple $(i,a,b) \in \{1,
\ldots, n\} \times D^2$ if $\pr_{1, \ldots, i-1}(t) = \pr_{1, \ldots,
 i-1}(t')$, $t[i] = a$, and $t'[i] = b$. 
The {\em signature} of an $n$-ary relation $R$ over $D$ is then
defined as  \[\sig(R) = \{(i,a,b) \in \{1, \ldots, n\} \times D^{2} \mid \exists
t,t' \in R \text{ such that } (t,t') \text{ witnesses } (i,a,b)\},\]
and we say that $R' \subseteq R$ is a {\em representation} of $R$ if $\sig(R) =
\sig(R')$.
If $R'$ is a representation of $R$ it is said to be {\em compact} if
$|R'| \leq 2|\sig(R)|$, and it is known that every relation invariant
under a Maltsev operation admits a compact
representation. Furthermore, we have the following theorem from
Bulatov and Dalmau, where we let $\cclone{R}_f$ denote the smallest
superset of $R$ preserved under the
operation $f$.

\begin{theorem}[\cite{bulatov2006b}] \label{theorem:kedge-rep}
  Let $\malt$ be a Maltsev operation over a finite domain, $R \in
  \inv(\{\malt\})$ a relation, and $R'$ a representation of $R$. Then
  $\cclone{R'}_{\malt} = R$.
\end{theorem}

Hence, relations invariant under Maltsev operations are
reconstructible from their representations.

\section{Maltsev Embeddings and Kernels of Linear Size}
\label{sec:embeddings}
In this section we give general upper bounds for kernelization of
NP-hard $\SAT$ problems based on algebraic conditions. We begin in
Section~\ref{sec:alg} by outlining the polynomial-time algorithm for Maltsev
constraints, and in Section~\ref{sec:upper} this algorithm is 
modified to construct linear-sized kernels for
certain $\SAT(\Gamma)$ problems.

\subsection{The Simple Algorithm for Maltsev Constraints}
\label{sec:alg}
At this stage the connection between Maltsev operations, compact
representations and tractability of Maltsev constraints might not be
immediate to the reader. We therefore give a brief
description of the simple algorithm for Maltsev constraints from
Bulatov and Dalmau~\cite{bulatov2006b}, which will henceforth simply
be referred to as the {\em Maltsev algorithm}. In a nutshell, the
algorithm operates as follows, where $\malt$ is a Maltsev operation
over a finite set $D$.

\begin{enumerate}
\item
  Let $(V, \{C_1, \ldots, C_m\})$ be an instance of
  $\CSP(\inv(\{\malt\}))$, and $S_0$ a compact representation of $D^{|V|}$.
\item
  For each $i \in \{1, \ldots, m\}$ compute a compact representation
  $S_i$ of the solution space of the instance $(V, \{C_1, \ldots, C_i\})$
  using $S_{i-1}$.
\item
  Answer yes if $S_m \neq \emptyset$ and no otherwise.
\end{enumerate}

The second step is accomplished by removing the tuples from
$\cclone{S_{i-1}}_{\malt}$ that are not compatible with the constraint
$C_i$. While the basic idea behind the Maltsev algorithm is not
complicated, the intricate details of the involved subprocedures are
outside the scope of this paper, and we refer the reader to Bulatov
and Dalmau~\cite{bulatov2006b}. 
We note that although the algorithm applies to infinite languages,
it is assumed that the relations in the input are specified by explicit
lists of tuples, i.e., the running time includes a factor proportional
to $\max |R|$ over relations $R$ used in the input. 

\begin{example} \label{ex:g1}
  Let $G = (D, \cdot)$ be a group over a finite set $D$, i.e., $\cdot$ is
  a binary, associative operator, $D$ is closed under $\cdot$ and
  contains an identity element $1_G$, and each element $x \in D$ has an
  inverse element $x^{-1} \in D$ such that $x \cdot x^{-1} = 1_G$. The ternary operation
  $s(x,y,z) = x \cdot y^{-1} \cdot z$ is referred to as the {\em coset
    generating operation} of $G$, and is Maltsev since
  $s(x,y,y) = x \cdot y^{-1} \cdot y = x$ and $s(x,x,y)
  = x \cdot x^{-1} \cdot y = y$. The problem $\CSP(\inv(\{s\}))$ is
  known to be tractable via the algorithm from Feder and
  Vardi~\cite{FV98}, but since $s$ is a Maltsev operation
  $\CSP(\inv(\{s\}))$ can equivalently well be solved via the Maltsev
  algorithm.
\end{example}

Another early class of tractable CSP problems was discovered via the
observation that if $R$ is preserved by a certain Maltsev operation,
it can be viewed as the solution space of a system of linear
equations.

\begin{example} \label{ex:g2}
  An {\em Abelian} group $G = (D, +)$ is a group where $+$ is
  commutative. Similar to Example~\ref{ex:g1} we can consider the
  coset generating operation $s(x,y,z) = x - y + z$, where $-y$ denotes the
  inverse of the element $y$. If $|D|$ is prime it is known that $R \in
  \inv(\{s\})$ if and only if $R$ is the solution space of a system of
  linear equations modulo $|D|$~\cite{jeavons1995}. Hence, the problem
  $\CSP(\inv(\{s\}))$ can efficiently be solved with Gaussian
  elimination, but can also be solved via the Maltsev algorithm.
\end{example}

\subsection{Upper Bounds Based on Maltsev Embeddings}
\label{sec:upper}
In this section we use a variation of the Maltsev algorithm to
obtain kernels of $\SAT$ problems. First, observe that
$\Gamma$ is never preserved by a Maltsev operation when $\SAT(\Gamma)$
is NP-hard~\cite{sch78}. However,
it is sometimes possible to find a related constraint
language $\hat \Gamma$ which is preserved by a Maltsev operation. This
will allow us to use the advantageous properties of relations
invariant under Maltsev operations in order to compute a kernel for
the original $\SAT(\Gamma)$ problem. We thus begin by making the
following definition.

\begin{definition} \label{def:embedding}
A constraint language $\Gamma$ over a domain $D$ admits an {\em
  embedding} over the constraint language $\hat \Gamma$ over $D' \supseteq D$ if there exists
a bijective function $h : \Gamma \rightarrow \hat \Gamma$ such that
$\ar(h(R)) = \ar(R)$ and $h(R) \cap D^{\ar(R)} = R$ for every $R \in \Gamma$.

\end{definition}

If $\hat \Gamma$ is preserved by a Maltsev operation then we say that
$\Gamma$ admits a {\em Maltsev embedding}. In general, we do not exclude the possibility that the domain $D'$ is
infinite. In this section, however, we will only be concerned with
finite domains, and therefore do not explicitly state this
assumption. 
If the bijection $h$ is efficiently computable and
there exists a polynomial $p$ such that $h(R)$ can be computed in
$O(p(|R|))$ time for each $R \in \Gamma$, then we say that $\Gamma$
admits a {\em polynomially bounded} embedding. In particular, an
embedding over a finite domain of any finite $\Gamma$ is
polynomially bounded. 

\begin{example} \label{ex:1}
  Recall from Section~\ref{sec:csp} that $\Roit$ consists of the three
  tuples $(0,0,1), (0,1,0)$, and $(1,0,0)$.  We claim that $\Roit$ has a
  Maltsev embedding over $\{0,1,2\}$. Let $\hroit = \{(x,y,z) \in
  \{0,1,2\}^3 \mid x + y + z = 1 \,(\bmod\, 3)\}$. By definition, $\hroit
  \cap \{0,1\}^3 = \Roit$, so all that remains to prove is that
  $\hroit$ is preserved by a Maltsev operation. But recall from
  Example~\ref{ex:g2} that a relation $R$ is the solution space of a system of linear equations
  over $D$, where $|D|$ is prime, if and only if $R$ is preserved by the 
  operation $x - y + z$ over $D$. Hence, $\hroit$ is indeed a Maltev
  embedding of $\Roit$. More generally, one can also prove that $\roik{k}$ has a Maltsev
  embedding to a finite domain $D$ where $|D| \geq k$ and $|D|$ is prime.
\end{example}

Given an instance $I = (\{x_1, \ldots, x_n\},C)$ of $\CSP(\Gamma)$ we
let $\rel{I} = \{(g(x_1), \ldots, g(x_n))
\mid g$ satisfies $I\}$.  If $\malt$ is a
Maltsev operation and $I = (V,\{C_1, \ldots, C_m\})$
an instance of $\CSP(\inv(\{\malt\}))$ we let $\seq(I) = (S_0, S_1 \ldots,
S_{m})$
 denote the
 compact representations of the relations $\rel{(V, \emptyset)},
 \rel{(V, \{C_1\})}, \ldots, \rel{(V, \{C_1,
 \ldots, C_m\})}$ computed by the Maltsev algorithm. We remark that
the ordering chosen in the sequence $\seq{I}$ does not influence the
upper bound in the forthcoming kernelization algorithm.

\begin{definition}
  Let $\malt$ be a Maltsev operation, $p$ a polynomial and let $\Delta \subseteq
  \inv(\{\malt\})$. We say that $\Delta$ and $\CSP(\Delta)$ have 
  {\em chain length $p$} if 
  $|\{\cclone{S_i}_{\malt} \mid i \in \{0, 1, \ldots, |C|\}\}| \leq p(|V|)$
  for each instance $I = (V,C)$ of $\CSP(\Delta)$,
where $\seq(I) = (S_0, S_1, 
  \ldots, S_{|C|})$.
\end{definition}

We now have everything in place to define our kernelization algorithm.

\begin{theorem} \label{thm:kernel}
  Let $\Gamma$ be a Boolean constraint language which admits a
  polynomially bounded Maltsev embedding $\hat \Gamma$ with chain length $p$. Then
  $\SAT(\Gamma)$ has a kernel with $O(p(|V|))$ constraints.
\end{theorem}

\begin{proof}
  Let $\malt \in \pol(\hat \Gamma)$ denote the Maltsev operation
witnessing the embedding $\hat \Gamma$. Given an instance $I = (V,C)$
of $\SAT(\Gamma)$ we can obtain an instance $I' = (V,C')$ of
$\CSP(\hat \Gamma)$ by replacing each constraint $R_i({\bf x_i})$ in
$C$ by $\hat R_i({\bf x_i})$. We arbitrarily order the constraints as
$C'=(C_1, \ldots, C_m)$ where $m=|C'|$. We then iteratively compute
the corresponding sequence $\seq(I') = (S_0, S_1, \ldots,
S_{|C'|})$. This can be done in polynomial time with respect to the
size of $I$ via the same procedure as the Maltsev
algorithm. For each $i \in \{1, \ldots, m\}$ we then do the
following.

  \begin{enumerate}
  \item
    Let the $i$th constraint be $C_i = \hat R_i(x_{i_1}, \ldots, x_{i_r})$ with $\ar(R_i)=r$.
  \item
    For each $t \in S_{i-1}$ determine whether $\pr_{i_1, \ldots, i_r}(t) \in \hat R_i$.
  \item
    If yes, then remove the constraint $C_i$, otherwise keep it.
  \end{enumerate}

  This can be done in polynomial time with respect to the size of the
  instance $I'$, since (1) $|S_{i-1}|$ is bounded by a polynomial in
  $|V|$ and (2) the test $\pr_{i_1, \ldots, i_r}(t) \in \hat R_i$ can naively be checked in linear time with
    respect to \smash{$|\hat R_i|$}. We claim that the procedure outlined above will correctly detect
  whether the constraint $C_i$ is redundant or not with respect to
  $\cclone{S_{i-1}}_{\malt}$, i.e., whether $\cclone{S_{i-1}}_{\malt} =
  \cclone{S_i}_{\malt}$. 
  First, observe that if there exists $t \in
  S_{i-1}$ such that \smash{$\pr_{i_1, \ldots, i_r}(t) \notin \hat R_i$},
  then the constraint is clearly not redundant. Hence, assume that
  $\pr_{i_1, \ldots, i_r}(t) \in \smash{\hat R_i}$ for every $t \in S_{i-1}$. 
  Then $S_{i-1} \subseteq \cclone{S_i}_{\malt}$, hence also 
  $\cclone{S_{i-1}}_{\malt} \subseteq \cclone{S_i}_{\malt}$. 
  On the other hand, $\cclone{S_i}_{\malt} \subseteq \cclone{S_{i-1}}_{\malt}$ 
  holds trivially.
  Therefore, equality must hold.

  Let $I'' = (V, C'')$ denote the resulting instance. Since
  $\CSP(\inv(\{\malt\}))$ has chain length $p$ it follows that (1) the
  sequence $\cclone{S_0}_{\malt}, \cclone{S_1}_{\malt}, \ldots, \cclone{S_{|C'|}}_{\malt}$
  contains at most $p(|V|)$ distinct elements, hence $|C''|\leq p(|V|)$,
  and (2) $\rel{I'} = \rel{I''}$. Clearly, it also holds that 
  $\rel{I} = (\rel{I'} \cap \{0,1\}^{|V|}) = (\rel{I''} \cap \{0,1\}^{|V|})$. 
  Hence, we can safely transform $I''$ to an instance $I^*$ of $\SAT(\Gamma)$ by
  replacing each constraint $\hat R_i({\bf x_i})$ with $R_i({\bf x_i})$. 
  Then $I^*$ is an instance of $\SAT(\Gamma)$ with at most $p(|V|)$ constraints,
  such that $\rel{I}=\rel{I^*}$.
  In particular, $I^*$ has a solution if and only if $I$ has a solution.
\end{proof}

Clearly, the above algorithm also works for finite-domain $\CSP$. As
with the Maltsev algorithm, the procedure runs in polynomial time with
respect to the total size of the instance. For languages with bounded
arity this simply means time polynomial in $n$, but it is worth noting
that if $\Gamma$ is infinite but somehow concisely encoded, then we
cannot necessarily check whether an $n$-ary constraint is redundant in
time polynomial in $n$. All that remains to be proven now is that there
actually exist Maltsev embeddings with bounded chain length.

\begin{definition}
  Let $f$ be an $n$-ary operation over $D$. A binary relation $R \in
  \inv(\{f\})$ is said to be a {\em congruence} of $f$ if it is an
  equivalence relation over $D$.
\end{definition}

Before we prove Theorem~\ref{thm:chainlength}, we need two subsidiary lemmas.

\begin{lemma} \label{lemma:sig}
  Let $\malt$ be a Maltsev operation over $D$ and $I$ an
  instance of $\CSP(\inv(\{\malt\}))$. Then $\sig(S_{i-1}) \supseteq
  \sig(S_{i})$ for each $S_{i-1}$ in $\seq(I)$.
 \end{lemma}

\begin{proof}
  Let $I = (V,C)$, $(j,a,b) \in \sig(S_i)$, where $j \in \{1,
  \ldots, |V|\}$ and $a,b \in D$. Then there exists $t,t' \in S_i$ 
  such that $(t,t')$ witnesses $(j,a,b)$, i.e., 
  $\pr_{1, \ldots, j - 1}(t) = \pr_{1, \ldots, j - 1}(t')$, and $t[j] = a$, 
  $t'[j] = b$. Since 
  $\cclone{S_{i-1}}_{\malt} \supseteq \cclone{S_{i}}_{\malt} \supseteq S_i$, 
  it follows that $t,t' \in \cclone{S_{i-1}}_{\malt}$, and hence also 
  that $(j,a,b) \in \sig(\cclone{S_{i-1}}_{\malt})$. But since $S_{i-1}$ is a
  representation of $\cclone{S_{i-1}}_{\malt}$, 
  $\sig(S_{i-1}) = \sig(\cclone{S_{i-1}}_{\malt})$, 
  from which we infer that $(j,a,b) \in \sig(S_{i-1})$.
\end{proof}

\begin{lemma} \label{lemma:eqclass}
  Let $\malt$ be a Maltsev operation over a finite domain $D$, and 
  $R \in \inv(\{\malt\})$. For every $i \in \{1,\ldots,\ar(R)\}$, the
  tuples $(i,a,b)$ in $\sig(R)$ define an equivalence relation on 
  $\pr_i(R) \subseteq D$.
\end{lemma}
\begin{proof} Define the relation $a \sim b$ if and only if $(i,a,b) \in \sig(R)$. 
  Note that $(i,a,a) \in \sig(R)$ if and only if $a \in \pr_i(R)$,
  and that $(i,a,b) \notin \sig(R)$ for any $b$ if $a \notin \pr_i(R)$. 
  Also note that $\sim$ is symmetric by its definition. It remains to show
  transitivity. Let $(i,a,b) \in \sig(R)$ be witnessed by $(t_a, t_b)$ and
  $(i,a,c) \in \sig(R)$ be witnessed by $(t_a', t_c')$. We claim that
  $t_c:=\malt(t_a,t_a',t_c') \in R$ is a tuple such that $(t_b, t_c)$
  witnesses $(i,b,c) \in \sig(R)$. Indeed, for every $j<i$ we have
  $\malt(t_a[j], t_a'[j], t_c'[j]) = \malt(t_a[j], t_a'[j], t_a'[j]) =
  t_a[j]$, whereas $\malt(t_a[i], t_a'[i], t_c'[i])=(a,a,c)=c$. Since
  $t_a[j]=t_b[j]$ for every $j<i$, it follows that $(t_b,t_c)$ witnesses
  $(i,b,c) \in \sig(R)$. Hence $\sim$ is an equivalence relation on $\pr_i(R)$.
\end{proof}

\begin{restatable}{theorem}{thmchainlength} \label{thm:chainlength}
  Let $\malt$ be a Maltsev operation over a finite domain $D$. Then 
  $\CSP(\inv(\{\malt\}))$ has chain length $O(|D||V|)$.
\end{restatable}

\begin{proof}
  Let $I=(V,C)$ be an instance of $\CSP(\inv(\{\malt\}))$, with
  $|V|=n$ and $|C|=m$, and let $\seq(I) = (S_0, S_1, \ldots, S_{m})$
  be the sequence of compact representations computed by the Maltsev algorithm.
  By Lemma~\ref{lemma:sig}, $\sig(S_{i+1}) \subseteq \sig(S_i)$ for every $i<m$, 
  and by Lemma~\ref{lemma:eqclass}, the sets $(j,a,b) \in \sig(S_i)$ induce
  an equivalence relation on $\pr_j(\cclone{S_i}_\malt)$ for every $i \leq m$, $j \leq n$. 
  (Lemma~\ref{lemma:eqclass} applies here since $\sig(S_i)=\sig(\cclone{S_i}_\malt)$
  for every $S_i$ in $\seq(I)$, and $\cclone{S_i}_\malt \in \inv(\{\malt\})$.)
  We also note that if $\sig(S_{i+1})=\sig(S_i)$, then
  $\cclone{S_i}_{\malt} = \cclone{S_{i+1}}_{\malt}$ since $S_{i+1}$ is a
  compact representation of $\cclone{S_i}_\malt$. 
  Hence, we need to bound the number of times that 
  $\sig(S_{i+1}) \subset \sig(S_i)$ can hold. 
  Now note that whenever $\sig(S_{i+1}) \subset \sig(S_i)$,
  then either $\pr_j(\cclone{S_i}_\malt) \subset \pr_j(\cclone{S_{i+1}}_\malt)$  
  for some $j$, or the equivalence relation induced by tuples 
  $(j,a,b) \in \sig(S_{i+1})$ is a refinement of that induced
  by tuples $(j,a,b) \in \sig(S_i)$ for some $j$. Both of these events
  can only occur $|D|-1$ times for every position $j$ (unless $S_m=\emptyset$).
  Hence the chain length is bounded by $2|V||D|$. 
\end{proof}

This bound can be slightly improved for a particular class of Maltsev
operations. Recall from Example~\ref{ex:g1} that $s(x,y,z) = x \cdot
y^{-1} \cdot z$ is the coset generating operation of a group $G = (D, \cdot)$.

\begin{lemma} \label{lemma:coset}
  Let $G = (D,\cdot)$ be a finite group and let $s$ be its coset
  generating operation. Then $\CSP(\inv(\{s\}))$ has chain length
  $O(|V|\log |D|)$. 
\end{lemma}

\begin{proof}
  Let $I = (V, C)$ be an instance of $\CSP(\inv(\{s\}))$, where $|V| =
  n$ and $|C| = m$. Let $\seq(I) = (S_0, S_1, \ldots, S_{m})$ be the
  corresponding sequence.  First observe that $S_0$ is a compact
  representation of $D^{n}$ and that $(D^{n}, \cdot)$ is nothing else
  than the $n$th direct power of $G$. It is well-known that $R$ is a
  coset of a subgroup of $(D^{n}, \cdot)$ if and only if $s$ preserves
  $R$~\cite{dalmau2003}. In particular, this implies that $S_1$ is a compact
  representation of a subgroup of $(D^n, \cdot)$, and more generally
  that each $S_i$ is a compact representation of a subgroup of
  $\cclone{S_{i-1}}_{s}$. An application of Lagrange's theorem reveals
  that $|\cclone{S_{i}}_s|$ divides $|\cclone{S_{i-1}}_s|$, which
  implies that the sequence $\cclone{S_0}_s, \cclone{S_1}_{s}, \ldots, \cclone{S_{m}}_s$
  contains at most $n \log_2 |D|+1$ distinct elements.
\end{proof}

Note that if the domain $|D|$ is prime in Lemma~\ref{lemma:coset} then
the proof can be strengthened to obtain the bound $O(|V|)$. As an application of this
result, let us briefly return to Example~\ref{ex:1}, where we
demonstrated that $\roik{k}$ had a Maltsev embedding over the coset
generating operation of an Abelian
group $(D,+)$ where $|D|$ is prime. Combining Theorem~\ref{thm:kernel} and
Lemma~\ref{lemma:coset} we therefore conclude that $\SAT(\{\roik{k}\})$ has a kernel with $O(|V|)$
constraints. More generally, we may interpret the results in this
section as follows. If $\Gamma$ admits a Maltsev embedding over the
coset generating operation of an Abelian group $(D,+)$, where $|D|$ is
prime, then we obtain kernels with $O(|V|)$ constraints, closely
mirroring the results from Jansen and Pieterse~\cite{JansenP16MFCS}.
This is in turn a special case of constraint languages admitting
Maltsev embeddings over coset generating operations over arbitrary
groups, where we obtain kernels with $O(|V|\log |D|)$ constraints. It
is not hard to find examples of groups whose coset generating
operations cannot be represented by the aforementioned Abelian
groups. One such example is the group $A_n$ of all even permutations
over $\{1, \ldots, n\}$ for $n \geq 3$. Last, in the most general
case, where we obtain kernels with $O(|V||D|)$ constraints, we have
embeddings over arbitrary Maltsev operations. Furthermore, it is known that a
Maltsev operation $\malt$ over $D$ is the coset generating operation
of a group $(D, \cdot)$ if and only if $\malt(\malt(x,y,z), z, u) =
\malt(x,y,u)$, $\malt(u, z, \malt(z,y,x)) = \malt(u,y,x)$ for all
$x,y,z,u \in D$~\cite{dalmau2003}. Hence, any Maltsev operation which
do not satisfy any of these two identities cannot be viewed as a coset
generating operation of some group.

\section{Kernels of Polynomial Size}
\label{sec:higher}
Section~\ref{sec:upper} gives a description of $\SAT$
problems admitting kernels with $O(n)$ constraints. In this section we
study two generalizations which provide kernels
with $O(n^c)$ constraints for $c > 1$.

\subsection{Moving Beyond Maltsev: $k$-Edge Embeddings}
\label{sec:edge}
It is known that Maltsev operations are particular examples of a more
general class of operations called {\em $k$-edge
  operations}. Following Idziak et al.~\cite{Idziak2010b} we define a
$k$-edge operation $e$ as a $(k+1)$-ary operation satisfying 
$e(x, x, y, y, y, \ldots, y, y) = e(x, y, x, y, y, \ldots, y, y) = y$
and for each $i \in \{4, \ldots, k+1\}$, $e(y,\ldots,y, x, y, \ldots,
y) = y$, where $x$ occurs in position $i$.
Note that a Maltsev operation is nothing else than a 2-edge operation
with the first and second arguments permuted. A {\em
$k$-edge embedding} is then defined analogously to the concept of a
Maltsev embedding, with the distinction that the embedding $\hat \Gamma$
must be preserved by a $k$-edge operation for some $k \geq 2$. It is
known that $k$-edge operations satisfy many of the advantageous
properties of Maltsev operations, and the basic definitions concerning
signatures and representations are similar. 
Before the proof of Theorem~\ref{thm:kedge} we need the following theorem
from Idziak et al.~\cite{Idziak2010b}.

\begin{theorem}~\cite{Idziak2010b} \label{theorem:edge_terms}
If $e$ is a $k$-edge operation over $D$ then $\clone{\{e\}}$ also
contains a binary operation $d$ and a ternary operation $p$ satisfying 
\[p(x, y, y) = x, p(x, x, y) = d(x,y), d(x,d(x,y)) = d(x,y),\]
and a $k$-ary operation $s$, satisfying
$s(x, y, y, y, \ldots, y, y) = d(y,x)$ and for each $i \in
\{2,\ldots,k\}$,
$s(y,y,\ldots, y,x,y,\ldots,y) = y$, where $x$ appears in position $i$.
\end{theorem}

If $e$ is a $k$-edge operation over $D$ and $d$ the operation in
Theorem~\ref{theorem:edge_terms} then $(a,b) \in D^2$ is a {\em
minority pair} if $d(a,b) = b$. Given an $n$-ary relation $R \in
\inv(\{e\})$ and $t,t' \in R$ we then say that the index $(i,a,b) \in
\{1,\ldots, n\} \times D^2$ {\em witnesses} $(t,t')$ if $(a,b)$ is a
minority pair and $\pr_{1, \ldots, i-1}(t) = \pr_{1, \ldots, i-1}(t')$
and $t[i] = a$, $t'[i] = b$. We let $\sig_e(R)$ denote the set of all
indexes witnessing tuples of the relation $R \in \inv(\{e\})$. Last,
$R' \subseteq R$ is a {\em representation} of $R$ if (1) $\sig_e(R) =
\sig_e(R')$ and (2) for every $i_1, \ldots, i_{k'} \in \{1, \ldots,
n\}$, $k' < k$, $\pr_{i_1, \ldots, i_{k'}}(R) = \pr_{i_1, \ldots,
i_{k'}}(R')$. Similar to the Maltsev case we have the following useful
property of representations of relations invariant under $k$-edge
operations.

\begin{theorem}~\cite{Idziak2010b} \label{theorem:kedge-rep2}
  Let $e$ be a $k$-edge operation over a finite domain, $R \in
  \inv(\{e\})$ a relation, and $R'$ a representation of $R$. Then
  $\cclone{R'}_e = R$.
\end{theorem}

Moreover, each $n$-ary relation invariant under a $k$-edge operation has a
compact representation of size $O(n^{k-1})$. By this stage it should not come as
a surprise to the reader that Maltsev algorithm outlined in
Section~\ref{sec:alg} can be modified to solve $\CSP(\inv(\{e\}))$ in
polynomial time. We will refer to this algorithm as the {\em few
  subpowers} algorithm~\cite{Idziak2010}. 
We then obtain analogous to the Maltsev case from
Section~\ref{sec:upper}. 

\begin{restatable}{theorem}{thmkedge} \label{thm:kedge}
  Let $\Gamma$ be a Boolean constraint language which admits a
  polynomially bounded $k$-edge embedding $\hat \Gamma$ over a finite
  domain $D$. Then $\SAT(\Gamma)$ has a kernel with
  $O(|D|^{k-1}|V|^{k-1})$ constraints.
\end{restatable}

\begin{proof} 
  We only provide a proof sketch since the details are
  very similar to the Maltsev case. Assume $k \geq 3$, since
  otherwise the bound follows from Theorem~\ref{thm:kernel}
  and~\ref{thm:chainlength}.
  Given an instance $I = (V,\{C_1,
  \ldots, C_m\})$, iteratively compute compact representations $S_0,
  S_1, \ldots, S_m$ of the solution space of $(V,\emptyset)$, $(V,
  \{C_1\})$, $\ldots$, $(V, \{C_1, \ldots, C_m\})$. This can be done
  in polynomial time using the procedures from the few subpowers
  algorithm~\cite{Idziak2010}. We then remove the
  constraint $C_i$ if and only if $\cclone{S_i}_e =
  \cclone{S_{i-1}}_e$. All that remains to be proven is therefore that
  $\{\cclone{S_0}_e, \cclone{S_1}_e, \ldots, \cclone{S_m}_e\}$ is
  bounded by $O(|D|^{k-1}|V|^{k-1})$. For each $S_i$ define
  \[ \text{Proj}(S_i) = \{(I,J) \mid I \in \{1,\ldots,|V|\}^i, J \in
  D^i, i < k, \pr_I(S_i) = J\}.
  \] If $\cclone{S_i}_e \supset \cclone{S_{i-1}}_e$ then either
  $\sig_e(S_{i}) \supset \sig_e(S_{i-1})$ or $\text{Proj}(S_i) \supset
  \text{Proj}(S_{i-1})$. This gives the bound
  $1+|\sig(D^n)|+|\textrm{Proj}(D^n)|=O(|D|^{k-1}|V|^{k-1})$.
\end{proof}

\subsection{Degree-$c$ Extensions}
\label{sec:poly}
We now consider an alternative technique for obtaining kernels with $O(n^c)$
constraints, $c>1$, which is useful for classes of languages that do not admit
Maltsev or $k$-edge embeddings.  
This will generalise the results on kernelization for
constraints defined via non-linear polynomials over finite
fields~\cite{JansenP16MFCS}.

\begin{definition} \label{def:degree}
We make the following definitions.
\begin{enumerate}
\item
Let $t \in \{0,1\}^r$ be a tuple of arity $r$
and $S_1, \ldots, S_{l}$ an enumeration of all subsets of $\{1,
\ldots, r\}$ of size at most $c$.  A tuple $\check t \in \{0,1\}^{l}$
is a {\em degree-$c$} extension of $t$ if $\check t[i]=\prod_{j \in S_i}
t[j]$, $i \in \{1, \ldots, l\}$.
\item
A \emph{degree-$c$ extension of $\Gamma$} is a language $\check \Gamma$
with a bijection $h$ between relations $R \in \Gamma$ and relations
$\check R \in \check \Gamma$ such that for every $R \in \Gamma$
and for every tuple $t \in \{0,1\}^{\ar(R)}$,
$t \in R$ if and only if $\check t \in \check R$ where $\check t$ is a
degree-$c$ extension of $t$.
\end{enumerate}
\end{definition}

Let $I=(V,C)$ be a $\SAT(\Gamma)$ instance for a Boolean constraint
language $\Gamma$. Let the \emph{degree-$c$ extension} of $V$ be the
set $V^{(c)}$ consisting of all subsets of $V$ of size at most $c$,
and from any assignment $g: V \to \{0,1\}$ we define an assignment
$g': V^{(c)} \to \{0,1\}$ as $g'(S):= \prod_{v \in S} g(v)$ for every
set $S \in V^{(c)}$. 
Degree-$c$ extensions, Maltsev embeddings and
$k$-edge embeddings are related by the following lemma. 

\begin{restatable}{theorem}{thmext} \label{thm:ext}
  Let $\Gamma$ be a finite Boolean language and $\check \Gamma$
  a degree-$c$ extension of $\Gamma$. If $\check \Gamma$ admits a Maltsev embedding,
  then $\SAT(\Gamma)$ admits a kernel of $O(n^c)$ constraints;
  if $\check \Gamma$ admits a $k$-edge embedding, then $\SAT(\Gamma)$
  admits a kernel of $O(n^{kc})$ constraints.
\end{restatable}

\begin{proof}
  Since $\Gamma$ is finite and fixed, we skirt all issues about how to
  compute the extension and the embedding. Let $I=(V,C)$, $|V| = n$, be
  an instance of $\SAT(\Gamma)$, and let $V^{(c)}$ be the degree-$c$
  extension of $V$. For each constraint $R(x_1, \ldots, x_m)$, $m =
  \ar(R)$, let $X_1, \ldots, X_l \in V^{(c)}$ denote the subsets of
  $\{x_1, \ldots, x_m\}$ of size at most $c$, and replace $R(x_1,
  \ldots, x_m)$ by the constraint $\check R(X_1, \ldots,
  X_l)$. 
  Let $I'$ be the 
  instance of $\SAT(\check \Gamma)$ resulting from repeating this for
  every constraint in the instance.  Observe that if $g$ is a satisfying assignment to $I$ then
  $g'(X) = \prod_{x \in X} g(x)$, $X \in V^{(c)}$, is a
  satisfying assignment to $I'$. We now apply the kernelization for
  languages with Maltsev embeddings, respectively $k$-edge embeddings,
  to $I'$, and let $I''=(V, C')$ where $C' \subseteq C$ is the set of
  constraints kept by the kernelization.  Note that the contents of the
  relation $\rel{I}$ defined by $I$ correspond directly to the relation
  $\{\check t \cap \rel{I'} \mid t \in \{0,1\}^{n}\}$.  Since the
  kernelizations we use preserve the entire solution space, this
  kernelization procedure is sound, and the desired bound for the number of
  constraints in the output follows.
\end{proof}

We observe that this captures the class of $\SAT$ problems which can be
written as roots of low-degree polynomials from Jansen and
Pieterse~\cite{JansenP16MFCS}. 

\begin{theorem}
  Let $\Gamma$ be a Boolean language such that every relation $R \in \Gamma$
  can be defined as the set of solutions in $\{0,1\}$ to a polynomial of 
  degree at most $d$, over some fixed finite field $F$. Then $\Gamma$
  admits a degree-$d$ extension with a Maltsev embedding. 
\end{theorem}
\begin{proof}
  We give a short sketch of the most important ideas. Let $G_1 =
  (D,\cdot)$ and $G_2 = (D, +)$ be the two Abelian groups representing
  the field $F$. For $R \in \Gamma$, let $p_R$ be the polynomial
  defining $R$.  Then $p_R$ can be written as a sum of monomials
  over $G_1$ of degree at most $d$, and each of these monomials
  corresponds to a member of $V^{(d)}$.  Hence, the extension $\check R$ of
  $R$ can be written as a linear sum over $G_2$, and similar to
  Example~\ref{ex:g2} it is now clear that the coset generating operation of $G_2$
  will preserve the resulting Maltsev embedding, and the result
  follows from Theorem~\ref{thm:ext}.
\end{proof}

\section{Universal Partial Maltsev Operations and Lower Bounds}
\label{section:lower_bounds}
We have seen that Maltsev embeddings and, more generally, $k$-edge
embeddings, provide an algebraic criterion for determining that a
$\SAT(\Gamma)$ problem admits a kernel of a fixed size. In this section we demonstrate
that our approach can also be used to give lower bounds for the
kernelization complexity of $\SAT(\Gamma)$. More specifically, we will
use the fact that if a satisfiability problem $\SAT(\Gamma)$ admits a
Maltsev embedding, then this can be witnessed by certain canonical
partial operations preserving $\Gamma$. We begin in
Section~\ref{section:universal} by studying properties of these
canonical partial operations, and in
Section~\ref{section:lower} prove that the absence of these
operations can be used to prove lower bounds on kernelizability. 

\subsection{Universal Partial Maltsev Operations}
\label{section:universal}
Let $f : D^{k} \rightarrow D$ be a $k$-ary operation over $D \supseteq
\{0,1\}$. We can
then in a natural way associate a partial Boolean operation $\rest{f}$ with $f$ by
restricting $f$ to the Boolean arguments which also result in a Boolean value. In
other words
$\domain(\rest{f}) = \{(x_1, \ldots, x_k) \in \{0,1\}^{k} \mid f(x_1,
\ldots, x_k) \in \{0,1\}\},$ and $\rest{f}(x_1, \ldots, x_k) = f(x_1,
\ldots, x_k)$ for every $(x_1, \ldots, x_k) \in
\domain(\rest{f})$. 
For Maltsev embeddings we have the
following straightforward lemma.

\begin{lemma} \label{lemma:partial_embedding}
  Let $\Gamma$ be a Boolean constraint language admitting a Maltsev
embedding $\hat \Gamma$. Then $\rest{f} \in \ppol(\Gamma)$ for every
$f \in \pol(\hat \Gamma)$.
\end{lemma}

\begin{proof}
  Assume, with the aim of reaching a contradiction, that
$\rest{f}(t_1, \ldots, t_n) \notin R$ for some $R \in \Gamma$ and some
$n$-ary $f \in \pol(\hat \Gamma)$. By construction, $\rest{f}(t_1, \ldots, t_n) = t$ is 
a Boolean tuple. But since $\hat R \cap \{0,1\}^{\ar(R)} = R$, this
implies (1) that $t \notin \hat R$ and (2) that $\rest{f}(t_1, \ldots,
t_n) = f(t_1, \ldots, t_n) = t \notin \hat R$. Hence, $f$ does not
preserve $\hat R$ or $\hat \Gamma$, and we conclude that $\rest{f} \in \ppol(\Gamma)$.
\end{proof}

A Boolean partial operation $f$ is a {\em universal partial Maltsev
operation} if $f \in \ppol(\Gamma)$ for every Boolean $\Gamma$ admitting a
Maltsev embedding. 

\begin{definition} \label{def:inf}
Let the infinite domain $\infd$ be recursively defined to contain $0$, $1$,
and ternary tuples of the form $(x,y,z)$ where $x,y,z \in
\infd$. The ternary Maltsev operation $u$ over $\infd$ is defined as
$u(x, x, y) = y, u(x, y, y) = x$, and $u(x, y, z) = (x,y,z)$
otherwise. 
\end{definition}
In the following theorem we show that if an operation $q$ is included in
the clone generated by the operation $u$,
then the partial Maltsev operation $\rest{q}$ is universal. 
Before the
presenting the proof we need some additional notation. It is
well-known that if $\clone{F}$ is a clone over a domain $D$ then $f
\in \clone{F}$ if and only if $f$ is definable as a term operation
over the algebra $(D, F)$~\cite{Goldstern2008}. Given a term $T(x_1, \ldots, x_n)$ over an
algebra $(D, F)$ defining a function $g \in \clone{F}$ and $b_1,
\ldots, b_n \in D$, we let $\val(T(b_1, \ldots, b_n)) = g(b_1, \ldots,
b_n)$.

\begin{restatable}{theorem}{thmuniversal} \label{thm:universal}
  Let $q \in \clone{\{u\}}$. Then $\rest{q}$ is a universal partial
  Maltsev operation.
\end{restatable}

\begin{proof}
Let $\Gamma$ be a Boolean constraint language which admits a Maltsev
embedding $\hat \Gamma$. We will prove that $\rest{q} \in
\ppol(\Gamma)$.  
  Let $p$ be the Maltsev operation witnessing the embedding $\hat
  \Gamma$, let $n$ denote the arity of $q$, and let $q(x_1, \ldots,
  x_{n}) = T^{u}(x_1, \ldots, x_n)$ where $T^{u}$ is the term over $u$
  defining $q$. 
  Now, first consider the operation $q' \in \clone{\{p\}}$ obtained by
  replacing each occurence of $u$ with $p$ in the term $T^{u}(x_1, \ldots,
  x_n)$. Let $T^{p}(x_1, \ldots, x_n)$ denote this term over $p$, and
  for each term $T^{u}_i(\mathbf{x_i})$ occurring as a subterm in $T^{u}(x_1, \ldots, x_n)$ we
  let $T^{p}_i(\mathbf{x_i})$ denote the corresponding term over $p$.
  
  Now observe that the partial operation $\rest{q'}$ is included in
  $\ppol(\Gamma)$ via Lemma~\ref{lemma:partial_embedding}. 
  We claim that $\rest{q}$ can be obtained as a subfunction of
  $\rest{q'}$, which implies that $\rest{q} \in \ppol(\Gamma)$, since
  a strong partial clone is always closed under taking
  subfunctions. By definition, we have that $(b_1, \ldots, b_n) \in
  \domain(\rest{q})$ if and only if $b_1, \ldots, b_n \in \{0,1\}$ and
  $q(b_1, \ldots, b_n) \in \{0,1\}$.

  We will prove that for each sequence of Boolean arguments $b_1,
  \ldots, b_n$, if $q(b_1, \ldots, b_n) = b \in \{0,1\}$  then
  $q'(b_1, \ldots, b_n) = b$. First, let us illustrate the intuition
  behind this by an example. Assume that $n = 7$ and that
  $T^u(x_1, x_2, x_3, x_4, x_5, x_6, x_7) = u(u(x_1, x_2, x_3), u(x_4, x_5,
  x_6), x_7)$. In this case we will e.g.\ have that
  $\val(T^u(0,1,0,0,1,0,1)) = 1$ since $u(u(0,1,0), u(0,1,0), 1) =
  u((0, 1, 0),(0, 1, 0), 1) = 1$, due to the fact that $u$ always
  respect the Maltsev identities. But since $p$ is also a Maltsev
  operation it must also be the case that $\val(T^p(0,1,0,0,1,0,1)) =
  1$, even if $u(0,1,0)$ and $p(0,1,0)$ might differ.
  
  We will now prove the general case by a case inspection of the term
  $T^{u}$. 
  First, assume that $T^{u}$ contains a term of the form $u(x_{i_1},
  x_{i_2}, x_{i_3})$. If $b_{i_1}, b_{i_2}, b_{i_3} \in \{0,1\}$ then
  $u(b_{i_1}, b_{i_2}, b_{i_3}) \in \{0,1\}$ if and only if $b_{i_1} =
  b_{i_2}$ or $b_{i_2} = b_{i_3}$. But this implies that $p(b_{i_1},
  b_{i_2}, b_{i_3}) = u(b_{i_1}, b_{i_2}, b_{i_3})$ since $p$ is
  Maltsev. Second, assume that $T^{u}$ contains a term of the form
  $u(T^u_1(\mathbf{x_1}), T^u_2(\mathbf{x_2}), T^u_3(\mathbf{x_3}))$
  where $\mathbf{x_1}, \mathbf{x_2}$ and $\mathbf{x_3}$ are tuples of
  variables over $x_1, \ldots, x_n$. Let $\mathbf{b_1}$, $\mathbf{b_2}$
  and $\mathbf{b_3}$ be Boolean tuples matching the length of
  $\mathbf{x_1}$, $\mathbf{x_2}$ and $\mathbf{x_3}$, and assume that
  $\val(T^u_1(\mathbf{b_1})) = \val(T^{p}_1(\mathbf{b_1}))$,
  $\val(T^u_2(\mathbf{b_2})) = \val(T^p_2(\mathbf{b_2}))$ and
  $\val(T^u_3(\mathbf{b_3})) = \val(T^p_3(\mathbf{b_3}))$. Similarly to
  the first case we have that
  \[u(\val(T^u_1(\mathbf{b_1})), \val(T^u_2(\mathbf{b_2})),
  \val(T^u_3(\mathbf{b_3}))) \in \{0,1\}\] if and only if
  $\val(T^u_1(\mathbf{b_1})) = \val(T^u_2(\mathbf{b_2}))$ or
  $\val(T^u_2(\mathbf{b_2})) = \val(T^u_3(\mathbf{b_3}))$, and since $p$
  is Maltsev this implies that \[p(\val(T^p_1(\mathbf{b_1})),
  \val(T^p_2(\mathbf{b_2})), \val(T^p_3(\mathbf{b_3}))) =
  u(\val(T^u_1(\mathbf{b_1})), \val(T^u_2(\mathbf{b_2})),
  \val(T^u_3(\mathbf{b_3}))).\]
  Hence, for each $(b_1, \ldots, b_n) \in \domain(\rest{q})$ we have
  that $(b_1, \ldots, b_n) \in \domain(\rest{q'})$ and that $\rest{q}(b_1, \ldots, b_n) = \rest{q'}(b_1,
  \ldots, b_n)$. This implies that $\rest{q}$ is a subfunction of
  $\rest{q'}$, that $\rest{q} \in \ppol(\Gamma)$, and, finally, that
  $\rest{q}$ is a universal partial Maltsev operation.
\end{proof}

Using Theorem~\ref{thm:universal} we can now prove that every Boolean
language $\Gamma$ invariant under the universal partial Maltsev
operations admits a Maltsev embedding over $\infd$. 

\begin{restatable}{theorem}{thminf} \label{thm:inf}
  Let $\Gamma$ be a Boolean constraint language. Then $\ppol(\Gamma)$ 
  contains all universal partial Maltsev
  operations if and only if $\Gamma$ has a Maltsev embedding $\hat \Gamma$ over
  $\infd$.
\end{restatable}

\begin{proof}
  For the first direction, let $u$ be the Maltsev operation from Definition~\ref{def:inf} over
the infinite domain $\infd$. For each relation $R \in \Gamma$ we let $\hat
R = \cclone{R}_u$. Let $\hat \Gamma$ denote the resulting constraint
language over $\infd$. By definition, $u \in \pol(\hat \Gamma)$, and
everything that remains to be proven is that $\hat R \cap
\{0,1\}^{\ar(R)} = R$ for each $\hat R \in \hat \Gamma$. Hence, assume
that there exists at least one tuple $t \in (\hat R \cap
\{0,1\}^{\ar(R)}) \setminus R$. This implies that there exists a term
$T$ over $u$ such that $\val(T(t_1[i], \ldots, t_m[i])) = t[i]$ for
each $i \in \{1, \ldots, \ar(R)\}$, where
$R = \{t_1, \ldots, t_m\}$. Let $q$ denote the function corresponding
to the term $T$ and observe that $q \in \clone{\{u\}}$.  According to
Theorem~\ref{thm:universal} this implies that $\rest{q}$ is a
universal partial Maltsev operation and, furthermore, that
$\rest{q}(t_1[i], \ldots, t_m[i])$ is defined for each $i \in \{1,
\ldots, \ar(R)\}$, since $q(t_1[i], \ldots, t_m[i]) \in
\{0,1\}$. Hence, $\rest{q}(t_1, \ldots, t_m) = t \notin R$, which
contradicts the assumption that $\Gamma$ was invariant under all
universal partial Maltsev operations. 

The second direction is trivial since if $\Gamma$ has a
Maltsev embedding over $\infd$ then $\Gamma$ by definition is
preserved by every universal partial Maltsev operation.
\end{proof}

It is worth remarking that Theorem~\ref{thm:universal} and
Theorem~\ref{thm:inf} implies that every universal partial Maltsev
operation can be described via Theorem~\ref{thm:universal}.

\subsection{Lower Bounds}
\label{section:lower}
Define the {\em first partial Maltsev operation} $\malt_1$ as
$\malt_1(x,y,y) = x$ and $\malt_1(x,x,y) = y$ for all $x,y \in
\{0,1\}$, and observe that $\domain(\malt_1) = \{(0,0,0), (1,1,1),
(0,0,1), (1,1,0), (1,0,0), (0,1,1)\}$.
Via Theorem~\ref{thm:universal} it follows that $\malt_1$ is
equivalent to $\rest{u}$, and is therefore a universal partial Maltsev
operation. In this section we will prove that $\malt_1 \in
\ppol(\Gamma)$ is in fact a necessary condition for the existence of a
linear-sized kernel for $\SAT(\Gamma)$, modulo a standard complexity
theoretical assumption. A pivotal part of this proof is
that if $\malt_1 \notin \ppol(\Gamma)$, then $\Gamma$ can qfpp-define
a relation $\gadget$, which can be used as a gadget in a
reduction from the \textsc{Vertex Cover} problem. This relation is
defined as $\gadget(x_1,x_2,x_3,x_4,x_5,x_6) \equiv (x_1 \lor x_4)
\land (x_1 \neq x_3) \land (x_2 \neq x_4) \land (x_5 = 0) \land (x_6 =
1)$. Note that the values enumerated by the arguments
of $\gadget$ is in a one-to-one correspondance with 
$\domain(\malt_1)$. However, as made clear in the following lemma, there is an even stronger relationship
between $\malt_1$ and $\gadget$.

\begin{restatable}{lemma}{lemmarfour} \label{lemma:r4}
If~$\Gamma$ is a Boolean constraint
  language such that $\cclone{\Gamma} = \br$ and $\malt_1 \notin \ppol(\Gamma)$ then
  $\gadget \in \pcclone{\Gamma}$.
\end{restatable}

\begin{proof}
  Before the proof we need two central observations. First, the
  assumption that $\cclone{\Gamma} = \br$ is well-known to be
  equivalent to that $\pol(\Gamma)$ consists only of
  projections~\cite{bsrv05}. Second, $\gadget$ consists of three
  tuples which can be ordered as $s_1, s_2, s_3$ in such a way that 
  there for every $s \in
  \domain(\malt_1)$ exists $1 \leq i \leq 6$ such that $s = (s_1[i],
  s_2[i], s_3[i])$.
  Now, assume that $\cclone{\Gamma} = \br$, $\malt_1 \notin \ppol(\Gamma)$,
  but that $\gadget \notin \pcclone{\Gamma}$. Then there exists an
  $n$-ary partial operation $f \in \ppol(\Gamma)$ such that $f \notin
  \ppol(\{\gadget\})$, and $t_1,
  \ldots, t_n \in \gadget$ such that $f(t_1, \ldots, t_n) \notin
  \gadget$. Now consider the value $k = |\{t_1, \ldots, t_n\}|$, i.e.,
  the number of distinct tuples in the sequence. If $n > k$ then it is known that there exists a closely related partial
  operation $g$ of arity at most $k$ such that $g \notin \ppol(\{\gadget\})$~\cite{lagerkvist2015}, and
  we may therefore assume that $n = k \leq |\gadget| = 3$. Assume first that $1 \leq n \leq
  2$. It is then not difficult to see that there for every $t \in
  \{0,1\}^n$ exists $i$ such that $(t_1[i], \ldots, t_n[i]) = t$. But
  then it follows that $f$ is in fact a total operation which is not a
  projection, which is impossible since we assumed that
  $\cclone{\Gamma} = \br$. Hence, it must be the case that $n = 3$,
  and that $\{t_1, t_2, t_3\} = \gadget$. Assume without loss of
  generality that $t_1 = s_1$, $t_2 = s_2$, $t_3 = s_3$, and note that
  this implies that $\domain(f) = \domain(\malt_1)$ (otherwise the
  arguments of $f$ can be described as a permutation of the arguments
  of $\malt_1$). First, we will show that $f(0,0,0) = 0$ and that
  $f(1,1,1) = 1$. Indeed, if $f(0,0,0) = 1$ or $f(1,1,1) = 0$, it is
  possible to define a unary total operation $f'$ as $f'(x) =
  f(x,x,x)$ which is not a projection since either $f'(0) = 1$ or
  $f'(1) = 0$. Second, assume there exists $(x,y,z) \in \domain(f)$,
  distinct from $(0,0,0)$ and $(1,1,1)$, such
  that $f(x,y,z) \neq \malt_1(x,y,z)$. Without loss of generality
  assume that $(x,y,z) = (a,a,b)$ for $a,b \in \{0,1\}$, and note that
  $f(a,a,b) = a$ since $\malt_1(a,a,b) = b$. If also $f(b,b,a) = a$ it is possible to define a binary
  total operation $f'(x,y) = f(x,x,y)$ which is not a projection,
  therefore we have that $f(b,b,a) = b$. We next consider the values
  taken by $f$ on the tuples $(b,a,a)$ and $(a,b,b)$. If $f(b,a,a) =
  f(a,b,b)$ then we can again define a total, binary operation which
  is not a projection, therefore it must hold that $f(b,a,a) \neq
  f(a,b,b)$. However, regardless of whether $f(b,a,a) = b$ or
  $f(b,a,a) = a$, it is not difficult to verify that $f$ must be a
  partial projection. This contradicts the assumption that $f \notin
  \ppol(\{\gadget\})$, and we conclude that $\gadget \in \pcclone{\Gamma}$.
\end{proof}

We will now use Lemma~\ref{lemma:r4} to
give a reduction from the \textsc{Vertex Cover} problem. It is known
that \textsc{Vertex Cover} 
does not admit a kernel with $O(n^{2 - \varepsilon})$ edges for any $\varepsilon > 0$, 
unless NP $\subseteq$ co-NP/poly~\cite{DellM14}. 
For each $n$ and $k$ let $H_{n,k}$ denote the relation 
$\{(b_1, \ldots, b_n) \in \{0,1\}^{n} \mid b_1 + \ldots + b_n = k\}$. 

\begin{lemma} \label{lemma:c}
  Let $\Gamma$ be a constraint language such that $\cclone{\Gamma} =
  \br$. Then $\Gamma$ can pp-define $H_{n,k}$ with $O(n+k)$ constraints
  and $O(n+k)$ existentially quantified variables.
\end{lemma}

\begin{proof}
  We first observe that one can recursively design a circuit consisting 
  of fan-in 2 gates which computes the sum of $n$ input gates as follows. At the
lowest level, we split the input gates into pairs and compute the sum
for each pair, producing an output of 2 bits for each pair. This can
clearly be done with $O(1)$ gates. At every level $i$ above that, we
join each pair of outputs from the previous level, of $i$ bits each,
into a single output of $i+1$ bits which computes their sum. This can
be done with $O(i)$ gates by chaining full adders. Finally, at level
$\lceil \log_2 n \rceil$, we will have computed the sum. The total
number of gates will be
\[
\sum_{i=1}^{\lceil \log_2 n \rceil} (\frac{n}{2^i}) \cdot O(i),
\]
and it is a straightforward exercise to show that this sums to $O(n)$.
Let $z_1, \ldots, z_{\log_2 n}$ denote the output
  gates of this circuit. By a standard Tseytin transformation we then
  obtain an equisatisfiable 3-SAT instance with $O(n)$ clauses and
  $O(n)$ variables~\cite{tseitin1983}. Next, for each $1 \leq i
  \leq \log_2 n$, add the unary
  constraint $(z_i = k_i)$, where $k_i$ denotes the $i$th bit of $k$
  written in binary. Each such unary constraint can clearly be
  pp-defined with $O(1)$ existentially quantified variables over $\Gamma$. We then pp-define each 3-SAT clause in order to
  obtain a pp-definition of $R$ over $\Gamma$, which in total only requires
  $O(n)$ existentially quantified variables. Note that this can be done
  since we assumed that $\cclone{\Gamma} = \br$ which implies that
  $\Gamma$ can pp-define every Boolean relation. 
\end{proof}

\begin{theorem} \label{theorem:col}
  Let $\Gamma$ be a finite Boolean constraint language such that
  $\cclone{\Gamma} = \br$ and $\malt \notin \ppol(\Gamma)$. Then
  $\SAT(\Gamma)$ does not have a kernel of size $O(n^{2 -
    \varepsilon})$ for any $\varepsilon > 0$, unless  NP $\subseteq$ co-NP/poly.
\end{theorem}

\begin{proof}
  We will give a polynomial-time many-one reduction from
  \textsc{Vertex Cover} parameterized by the number of vertices to
  $\SAT(\Gamma \cup \{\gadget\})$, which via
  Theorem~\ref{theorem:ppol_red} and Lemma~\ref{lemma:r4}
  has a reduction to $\SAT(\Gamma)$ which does not
  increase the number of variables. Let $(V,E)$ be
  the input graph and let $k$ denote the maximum size of the cover. 
  First, introduce two fresh variables $x_v$ and $x'_v$ for each $v
  \in V$, and one variable $y_i$ for each $1 \leq i \leq
  k$. Furthermore, introduce two variables $x$ and $y$. For each edge
  $\{u,v\} \in E$ introduce a constraint $\gadget(x_u, x'_v, x'_u, x_v,
  x, y)$, and note that this enforces the constraint $(x_u \vee x_v)$.
  Let $\exists z_1, \ldots, z_m . \phi(x_1, \ldots, x_{|V|}, y_1,
  \ldots, y_{k}, z_1, \ldots, z_m)$ denote the pp-definition $H_{|V| +
    k,k}$ over $\Gamma$ where $m \in O(k + |V|)$, and consisting of at
  most $O(k + |V|)$ constraints. Such a pp-definition must exist
  according to Lemma~\ref{lemma:c}. Drop the
  existential quantifiers and add the constraints of 
  $\phi(x_1, \ldots, x_{|V|}, y_1, \ldots, y_k, z_1, \ldots, z_m)$.
  Let $(V', C)$ denote this instance of $\SAT(\Gamma \cup \{\gadget\})$. Assume first
  that $(V,E)$ has a vertex cover of size $k' \leq k$. We first assign
  $x$ the value 0 and $y$ the value 1. For each $v$ in this cover assign
  $x_v$ the value 1 and $x'_v$ the value 0. For any vertex not included
  in the cover we use the opposite values. We then set $y_1,
  \ldots, y_{k - k'}$ to 1, and $y_{k - k' + 1}, \ldots, y_{k}$ to 0.
  For the other direction, assume that $(V', C)$ is satisfiable. For any
  $x_v$ variable assigned 1 we then let $v$ be part of the vertex
  cover. Since $x_1 + \ldots + x_{|V|} + y_1 + \ldots + y_k = k$, 
  the resulting vertex cover is smaller than or equal to $k$. 
\end{proof}

As an example, let $R^k = \{(b_1, \ldots, b_k) \in
\{0,1\}^{k} \mid b_1 + \ldots + b_k \in \{1,2\}\ \,(\bmod
\,6)\}$ and let $P = \{R^k \mid k \geq 1\}$. The kernelization status of $\SAT(P)$ was left open in Jansen
and Pieterse~\cite{JansenP16MFCS}, and while a precise upper bound
seems difficult to obtain, we can at least prove that this problem
does not admit a kernel of linear size, unless NP $\subseteq$
co-NP/poly. To see this, observe that
$(0,0,1),(0,1,1),(0,1,0) \in R^3$ but $\malt_1((0,0,1),(0,1,1),(0,1,0))
= (0,0,0) \notin R^3$. The result then follows from
Theorem~\ref{theorem:col}.

At this stage, it might be tempting to conjecture that $\malt_1
\in \ppol(\Gamma)$ is also a sufficient condition for a Maltsev
embedding, and, more ambitiously, that this is also a sufficient
condition for a kernel with $O(n)$ constraints. 
We can immediately rule out the first possibility by exhibiting a
relation $R$ and a
universal partial Maltsev operation $\malt$ such that $R$ is invariant
under $\malt_1$ but not under $\malt$. For example, first define the
operation $\malt_2$ as $\rest{q}$ where $q(x_1, x_2, x_3, x_4, x_5,
x_6, x_7, x_8, x_9) = u(u(x_1, x_2, x_3), u(x_4, x_5, x_6), u(x_7,
x_8, x_9))$. Second, consider the relation $R$ of arity $|\domain(\malt_2)|$ consisting of 9 tuples
$t_1, \ldots, t_9$ such that there for each $t \in \domain(\malt_2)$
exists exactly one $1 \leq i \leq
|\domain(\malt_2)|$ such that $(t_1[i], \ldots, t_9[i]) = t$. Then, by
definition, $\malt_2(t_1, \ldots, t_9) \notin R$ (since $\malt_2$ is
not a partial projejction), and hence does not preserve $R$, but one can verify that
this relation is preserved by $\malt_1$.
To rule out the second
possibility we have to find a constraint language $\Gamma$ such that
$\malt_1$ preserves $\Gamma$ but $\SAT(\Gamma)$ does not have a kernel
with $O(n)$ constraints. In fact, we will prove a stronger result, and
show that whenever $P$ is a finite set of partial operations such that
$\cclone{\inv(P)} = \br$ (and thus, $\SAT(\inv(P))$ is NP-complete),
then $\SAT(\inv(P))$ does not admit a polynomial kernel, unless NP $\subseteq$ co-NP/poly.
This is in contrast to the existing parameterized dichotomy results for CSP~\cite{bulatov14,KratschMW16,KratschW10,KrokhinM12,Marx05},
but as noted in the introduction,  
it is an expected conclusion when the parameter is~$n$; cf.~\cite[Lemma~35]{lagerkvist2016b}.

\begin{restatable}{theorem}{theoremnofinite} \label{theorem:nofin}
  Let $P$ be a finite set of partial polymorphisms
  such that $\cclone{\inv(P)}=\br$. Then
  $\SAT(\inv(P))$ does not admit a polynomial kernel
  unless NP $\subseteq$ co-NP/poly.
\end{restatable}

\begin{proof}
  We will show a reduction from $k$-SAT on $n$ variables
  to $\SAT(\inv(P))$ on $O(n^c)$ variables for some absolute
  constant $c$ that only depends on $P$. 
  Since $k$-SAT admits no kernel of size $O(n^{k-\varepsilon})$
  for any $\varepsilon>0$ unless NP $\subseteq$ co-NP/poly~\cite{DellM14},
  and since $c$ is independent of $k$, the result will follow. 

  The strategy to show the reduction is similar
  to that of Lemma~35 of~\cite{lagerkvist2016b}).
  Let $c$ be the largest arity of a partial polymorphism in $P$, 
  and let $(X,C)$ be an instance of $k$-SAT, $|X|=n$.
  Create a set of padding variables 
  $Y=\{y_{\bar x,f} \mid \bar x \in X^d, f : \{0,1\}^d \to \{0,1\}\}$,
  one for every $(x_1,\ldots,x_d) \in X^d$ and every $d$-ary function,
  $d=c^2$.  We will constrain so that for every $y_{\bar x, f} \in Y$ and 
  every satisfying assignment, we have $y_{\bar x, f}=f(\bar x(1), \ldots, \bar x(d))$. 
  The following is a central observation. 
  
  \begin{claim}
    Let $R \subseteq \{0,1\}^r$ be any $r$-ary Boolean relation,
    and let $X=\{x_1,\ldots,x_r\}$. Let $Y$ be a set of padding
    variables for $X$ as above. Define a relation $R'$ by
    \[
    R'(X,Y) \equiv R(X) \land \bigwedge_{y_{\bar x, f} \in Y} (y_{\bar x, f} = f(\bar x(1), \ldots, \bar x(c))).
    \]  
    Then $R(X) \equiv \exists Y R'(X,Y)$ and $R'$ is invariant under
    every non-total partial operation of arity at most $c$.
  \end{claim}
  \begin{proof}
  Let $\phi$ be a non-total partial operation of arity $c' \leq c$, 
  and assume that $R'$ is not invariant under $\phi$, i.e.,
  there are tuples $t_1, \ldots, t_{c'} \in R'$ such that
  $\phi(t_1,\ldots,t_{c'})$ is defined and not contained in $R'$.
  We assume that all tuples $t_i$ are distinct, as
  otherwise the application $\phi(t_1,\ldots,t_c)$ defines an operation $\phi'$ 
  of arity $|\{t_1,\ldots,t_{c'}\}|$ for which we can repeat the argument below.
  Let $u_1, \ldots, u_{c'}$ be the projections of the tuples onto $X$.
  Note that the tuples $u_1, \ldots, u_{c'}$ are distinct,
  and that $\phi(u_1, \ldots, u_{c'})$ is defined.
  Let $I \subseteq [r]$ be a minimal set of ``witness positions'' for the distinctness of $u$,
  i.e., for every pair $i, j \in [c']$, $i \neq j$, there is a position $p \in I$
  such that $u_i[p] \neq u_j[p]$. Note that $|I| \leq c^2$.
  Let 
  $t \in \{0,1\}^{c'}$ be a tuple for which $\phi$ is undefined.
  Then there exists a function $f: \{0,1\}^{|I|} \to \{0,1\}$ such that
  $f(\pr_I u_i)=t[i]$ for each $i \in [c']$, 
  since the projection onto $I$ is distinct for all tuples $u_i$, 
  and since $|I| \leq d^2$, there exist a  variable $y_{\bar x, f}$ in $Y$. 
  This implies that $\phi(t_1,\ldots,t_c)$ is undefined,
  since in particular $\phi$ is undefined when applied to the position corresponding
  to $y_{\bar x, f}$.
  Since $\phi$ was generically chosen, the claim follows.
\end{proof}

We can now wrap up the proof as follows. 
By Theorem~\ref{theorem:pgalois},
the language $\inv(P)$ has a quantifier-free pp-definition of any relation $R$
such that $R$ is invariant under any partial operation in $P$. 
By the above claim, this in particular includes the padded
relation $R'$ for any given relation $R$.
Now, if $(X,C)$ is an instance of $k$-SAT,  with $|X|=n$ as above,
and if $Y$ is the set of padding variables,
then we can output an instance of $\SAT(\inv(P))$
by replacing every $k$-clause in the input, defining a relation $R(V)$
for some $V \subseteq X$,
by the relation $R'(V, Y_V)$ according to the above claim.
Note in particular that the padding variables $Y_V$ used in this reduction
can be chosen from the set $X_V$. 
Finally, since $k$ is constant, the relations $R'(V, Y_V)$ 
and hence the output
can be enumerated in polynomial time.
\end{proof}

\section{Concluding Remarks and Future Research} 
We have studied the kernelization properties of
$\SAT$ and $\CSP$ problems parameterized by the number of variables with tools from
universal algebra. We particularly focused on problems with linear kernels,
and showed that a $\CSP$ problem has a kernel with $O(n)$ constraints if
it can be embedded into a $\CSP$ problem preserved by a Maltsev
operation; thus extending previous results in this direction. 
On the other hand, we showed that a $\SAT$ problem not preserved
by a partial Maltsev polymorphism does not admit such a kernel, 
unless NP $\subseteq$ co-NP/poly. 
This shows that the algebraic approach is viable for studying such 
fine-grained kernelizability questions.
More generally, we also gave 
algebraic conditions for the existence of a kernel with $O(n^c)$ constraints, $c>1$,
generalising previous results on kernels for $\SAT$ problems defined via
low-degree polynomials over a finite field.
Our work opens several directions for future research. 

\paragraph{A dichotomy theorem for linear kernels?} Our results suggest a
possible dichotomy theorem for the existence of linear kernels for $\SAT$ 
problems, at least for finite languages. However, two gaps remain towards
such a result. On the one hand, we have proven that any language $\Gamma$
preserved by the universal partial Maltsev operations admits a Maltsev
embedding over an infinite domain. However, the kernelization algorithms
only work for languages with Maltsev embeddings over finite domains. 
Does the existence of an infinite-domain Maltsev embedding 
for a finite language imply the existence of a Maltsev embedding over
a finite domain? Alternatively, can the algorithms be adjusted to work
also for languages with infinite domains, given that this domain is 
finitely generated in a simple way?
On the other hand, we only have necessity results
for the first partial operation $\malt_1$ out of an infinite set of
conditions for the positive results. Is it true that every universal
partial Maltsev operation is a partial polymorphism of every language with a
linear kernel, or do there exist $\SAT$ problems with linear kernels
that do not admit Maltsev embeddings?

\paragraph{Cases of higher degree.} Compared to the case of linear kernels,
our results on kernels with $O(n^c)$ constraints, $c>1$,
are more partial. Does the combination of degree extensions and $k$-edge
embeddings cover all cases, or are there further $\SAT$ problems
with non-trivial polynomial kernel bounds to be found?

\paragraph{The Algebraic CSP Dichotomy Conjecture.} 
Several solutions to the CSP dichotomy conjecture have
been announced~\cite{bulatov2017,feder2017,zhuk2017}. If correct,
these algorithms solve $\CSP(\Gamma)$ in polynomial time whenever
$\Gamma$ is preserved by a {\em Taylor term}. One can then define the concept of a Taylor
embedding, which raises the question of whether the proposed algorithms
can be modified to construct polynomial kernels. More generally, when
can an operation $f$ such that $\CSP(\inv(\{f\}))$ is tractable be
used to construct improved kernels? On the
one hand, one can prove that {\em $k$-edge operations},
which are generalized Maltsev operations, can be
used to construct kernels with $O(n^{k-1})$ constraints via a variant of the {\em
  few subpowers algorithm}. On the other hand, it is known that
relations invariant under  {\em semilattice operations} can
be described as generalized Horn formulas, but it is not evident how
this property could be useful in a kernelization procedure. 

\bibliography{references_kernelization}

\begin{thebibliography}{10}

\bibitem{barto2014}
L.~Barto.
\newblock Constraint satisfaction problem and universal algebra.
\newblock {\em ACM SIGLOG News}, 1(2):14--24, October 2014.

\bibitem{Idziak2010b}
J.~Berman, P.~Idziak, P.~Markovic, R.~McKenzie, M.~Valeriote, and R.~Willard.
\newblock Varieties with few subalgebras of powers.
\newblock {\em Transactions of the American Mathematical Society},
  362(3):1445--1473, 2010.

\bibitem{BKKR69i}
V.~G. Bodnarchuk, L.~A. Kaluzhnin, V.~N. Kotov, and B.~A. Romov.
\newblock Galois theory for {P}ost algebras. {I}.
\newblock {\em Cybernetics}, 5:243--252, 1969.

\bibitem{BKKR69ii}
V.~G. Bodnarchuk, L.~A. Kaluzhnin, V.~N. Kotov, and B.~A. Romov.
\newblock Galois theory for {P}ost algebras. {I}{I}.
\newblock {\em Cybernetics}, 5:531--539, 1969.

\bibitem{bsrv05}
E.~B\"ohler, H.~Schnoor, S.~Reith, and H.~Vollmer.
\newblock Bases for {B}oolean co-clones.
\newblock {\em Information Processing Letters}, 96(2):59--66, 2005.

\bibitem{bulatov2017}
A.~Bulatov.
\newblock A dichotomy theorem for nonuniform {CSP}s.
\newblock {\em CoRR}, abs/1703.03021, 2017.

\bibitem{bulatov2006b}
A.~Bulatov and V.~Dalmau.
\newblock A simple algorithm for {M}al{\textquoteright}tsev constraints.
\newblock {\em SICOMP}, 36(1):16--27, 2006.

\bibitem{bulatov2005}
A.~Bulatov, P.~Jeavons, and A.~Krokhin.
\newblock Classifying the complexity of constraints using finite algebras.
\newblock {\em SICOMP}, 34(3):720--742, March 2005.

\bibitem{bulatov14}
A.~Bulatov and D.~Marx.
\newblock Constraint satisfaction parameterized by solution size.
\newblock {\em {SIAM} Journal on Computing}, 43(2):573--616, 2014.

\bibitem{creignou2008b}
N.~Creignou and H.~Vollmer.
\newblock Boolean constraint satisfaction problems: When does {P}ost's lattice
  help?
\newblock In N.~Creignou, P.~G. Kolaitis, and H.~Vollmer, editors, {\em
  Complexity of Constraints}, volume 5250 of {\em Lecture Notes in Computer
  Science}, pages 3--37. Springer Berlin Heidelberg, 2008.

\bibitem{dalmau2003}
V.~Dalmau and P.~Jeavons.
\newblock Learnability of quantified formulas.
\newblock {\em TCS}, 306(1–3):485 -- 511, 2003.

\bibitem{DellM14}
H.~Dell and D.~van Melkebeek.
\newblock Satisfiability allows no nontrivial sparsification unless the
  polynomial-time hierarchy collapses.
\newblock {\em J. {ACM}}, 61(4):23:1--23:27, 2014.

\bibitem{FV98}
T.~Feder and M.~Vardi.
\newblock The computational structure of monotone monadic {SNP} and constraint
  satisfaction: A study through datalog and group theory.
\newblock {\em SICOMP}, 28(1):57--104, 1998.

\bibitem{Gei68}
D.~Geiger.
\newblock Closed systems of functions and predicates.
\newblock {\em Pac. J. Math.}, 27(1):95--100, 1968.

\bibitem{Goldstern2008}
M.~Goldstern and M.~Pinsker.
\newblock A survey of clones on infinite sets.
\newblock {\em Algebra universalis}, 59(3):365--403, 2008.

\bibitem{Idziak2010}
P.~Idziak, P.~Markovi\'{c}, R.~McKenzie, M.~Valeriote, and R.~Willard.
\newblock Tractability and learnability arising from algebras with few
  subpowers.
\newblock {\em SIAM Journal on Computing}, 39(7):3023--3037, June 2010.

\bibitem{impagliazzo98}
R.~Impagliazzo, R.~Paturi, and F.~Zane.
\newblock Which problems have strongly exponential complexity?
\newblock {\em Journal of Computer and System Sciences}, 63:512--530, 2001.

\bibitem{JansenP16MFCS}
B.~M.~P. Jansen and A.~Pieterse.
\newblock Optimal sparsification for some binary {CSPs} using low-degree
  polynomials.
\newblock In {\em Proceedings of MFCS 2016}, volume~58, pages 71:1--71:14.

\bibitem{jansen15}
B.~M.~P. Jansen and A.~Pieterse.
\newblock Sparsification upper and lower bounds for graphs problems and
  not-all-equal {SAT}.
\newblock In {\em Proceedings of {IPEC 2015}, Patras, Greece}.

\bibitem{jeavons1998}
P.~Jeavons.
\newblock On the algebraic structure of combinatorial problems.
\newblock {\em TCS}, 200:185--204, 1998.

\bibitem{jeavons1995}
P.~Jeavons, D.~Cohen, and M.~Gyssens.
\newblock A unifying framework for tractable constraints.
\newblock In {\em Proceedings of CP 1995}, pages 276--291.

\bibitem{jonsson2017}
P.~Jonsson, V.~Lagerkvist, G.~Nordh, and B.~Zanuttini.
\newblock Strong partial clones and the time complexity of {SAT} problems.
\newblock {\em JCSS}, 84:52 -- 78, 2017.

\bibitem{KratschMW16}
S.~Kratsch, D.~Marx, and M.~Wahlstr{\"{o}}m.
\newblock Parameterized complexity and kernelizability of max ones and exact
  ones problems.
\newblock {\em {TOCT}}, 8(1):1, 2016.

\bibitem{KratschW10}
S.~Kratsch and M.~Wahlstr{\"{o}}m.
\newblock Preprocessing of min ones problems: {A} dichotomy.
\newblock In {\em {ICALP} {(1)}}, volume 6198 of {\em Lecture Notes in Computer
  Science}, pages 653--665. Springer, 2010.

\bibitem{KrokhinM12}
A.~A. Krokhin and D.~Marx.
\newblock On the hardness of losing weight.
\newblock {\em {ACM} Trans. Algorithms}, 8(2):19, 2012.

\bibitem{lagerkvist2016b}
V.~Lagerkvist and M.~Wahlstr\"om.
\newblock The power of primitive positive definitions with polynomially many
  variables.
\newblock {\em JLC}, 2016.

\bibitem{lagerkvist2015}
V.~Lagerkvist, M.~Wahlstr{\"{o}}m, and B.~Zanuttini.
\newblock Bounded bases of strong partial clones.
\newblock In {\em Proceedings of the ISMVL 2015}.

\bibitem{Marx05}
D.~Marx.
\newblock Parameterized complexity of constraint satisfaction problems.
\newblock {\em Comput. Complexity}, 14(2):153--183, 2005.

\bibitem{nemhauser1975}
G.~L. Nemhauser and L.~E. Trotter.
\newblock Vertex packings: Structural properties and algorithms.
\newblock {\em Math. Programming}, 8(1):232--248, 1975.

\bibitem{feder2017}
A.~Rafiey, J.~Kinne, and T.~Feder.
\newblock Dichotomy for digraph homomorphism problems.
\newblock {\em CoRR}, abs/1701.02409, 2017.

\bibitem{romov1981}
B.A. Romov.
\newblock The algebras of partial functions and their invariants.
\newblock {\em Cybernetics}, 17(2):157--167, 1981.

\bibitem{sch78}
T.~Schaefer.
\newblock The complexity of satisfiability problems.
\newblock In {\em Proceedings of the 10th Annual ACM Symposium on Theory Of
  Computing (STOC-78)}, pages 216--226. ACM Press, 1978.

\bibitem{tseitin1983}
G.~S. Tseitin.
\newblock {\em Automation of Reasoning: 2: Classical Papers on Computational
  Logic 1967--1970}, chapter On the Complexity of Derivation in Propositional
  Calculus, pages 466--483.
\newblock Springer Berlin Heidelberg, Berlin, Heidelberg, 1983.

\bibitem{zhuk2017}
D.~Zhuk.
\newblock The proof of {CSP} dichotomy conjecture.
\newblock {\em CoRR}, abs/1704.01914, 2017.

\end{thebibliography}
\bibliographystyle{plain}
\end{document}